%% file: paper.tex
\documentclass[prodmode,acmjacm,english]{acmsmall}
\usepackage[T1]{fontenc}
\usepackage[latin9]{inputenc}
\usepackage{amssymb}
\usepackage{amsmath}
\usepackage{tikz}
\usetikzlibrary{intersections}
\usetikzlibrary{calc}

\input{macros}

\usepackage[ruled]{algorithm2e}

\SetAlFnt{\small}
\SetAlCapFnt{\small}
\SetAlCapNameFnt{\small}
\SetAlCapHSkip{0pt}
\IncMargin{-\parindent}

\acmVolume{0}
\acmNumber{0}
\acmArticle{0}
\acmYear{0}
\acmMonth{0}

\begin{document}

\setcopyright{acmlicensed}
\issn{0004-5411/2016} 
\acmYear{2016} 
\acmMonth{0} 
\acmArticle{0}
\doi{http://dx.doi.org/10.1145/2857050}

\markboth{V. Chonev et al.}{On the Complexity of the Orbit Problem}

\title{On the Complexity of the Orbit Problem}
\author{VENTSISLAV CHONEV
\affil{Institute of Science and Technology Austria}
JO{\"E}L OUAKNINE
\affil{University of Oxford}
JAMES WORRELL
\affil{University of Oxford}}

\input{abstract}




\acmformat{Ventsislav Chonev, Jo{\"e}l Ouaknine, and James Worrell, 2014.
On the Complexity of the Orbit Problem.}

\begin{bottomstuff}
Authors' addresses: Ventsislav Chonev, IST Austria, Am Campus 1,
3400 Klosterneuburg, Austria\\
\{Jo{\"e}l Ouaknine, James Worrell\}, 
Department of Computer Science, University of Oxford,
Wolfson Building, Parks Road, Oxford OX1 3QD, UK
\end{bottomstuff}

\maketitle

\input{secIntroduction}

\input{orbit.secOutline}

\input{orbit.secReduction}

\input{orbit.secDim1}

\input{orbit.secDim2}
\input{orbit.secDim3}

\appendixhead{CHONEV}

\received{September 2015}{}{}

\appendix

\input{math.secNT}

\input{math.secLRS}

\input{skolem.secIntro}

\input{skolem.secOutline}
\input{skolem.secOrder2}
\input{skolem.secBaker}

\input{skolem.secOrder3}

\input{skolem.secOrder4}

\medskip

\bibliographystyle{ACM-Reference-Format-Journals}
\bibliography{thesis}

\end{document}

%% file: macros.tex
\newcommand{\alg}{\mathbb{A}} 
\newcommand{\ra}{\mathbb{R}\cap\mathbb{A}} 
\newcommand{\re}{\mathbb{R}} 
\newcommand{\cplx}{\mathbb{C}} 
\newcommand{\nat}{\mathbb{N}} 
\newcommand{\zed}{\mathbb{Z}} 
\newcommand{\rat}{\mathbb{Q}} 
\newcommand{\rats}{\rat} 
\newcommand{\field}{\mathbb{F}} 

\newcommand{\defn}{\stackrel{\mathrm{def}}{=}} 

\newcommand{\plh}{\mathcal} 
\newcommand{\lrs}[1]{\langle #1_n\rangle_{n=0}^{\infty}} 
\newcommand{\vsp}{\plh} 

\newcommand{\matr}{\boldsymbol} 
\newcommand{\vct}{\boldsymbol} 

\newcommand{\cclass}[1]{\mathbf{#1}}
\newcommand{\np}{\cclass{NP}}
\newcommand{\rp}{\cclass{RP}}
\newcommand{\nprp}{\np^{\rp}}
\newcommand{\pspace}{\cclass{PSPACE}}
\newcommand{\ptime}{\cclass{PTIME}}

\newcommand{\eqslp}{\cclass{EqSLP}}

\newcommand{\npeqslp}{\np^{\eqslp}}

\newcommand{\corp}{\cclass{coRP}}

\newcommand{\bigoh}{\mathcal{O}}
\newcommand{\ein}[1]{2^{\bigoh(||#1||)}}
\newcommand{\pin}[1]{||#1||^{\bigoh(1)}}

\newcommand{\len}[1]{\Vert #1\Vert}

\newcommand{\eq}{\mathit{eq}}
\newcommand{\Eq}{\mathit{Eq}}

\newcommand{\cls}{}
\newcommand{\mul}{\mathit{mul}}
\newcommand{\spa}{\mathit{span}}
\newcommand{\orbitcoeff}{\kappa}
\newcommand{\orbitphi}{c}


%% file: abstract.tex
\begin{abstract}

We consider higher-dimensional versions of Kannan and Lipton's
Orbit Problem---determining whether a target vector space $\vsp{V}$ may
be reached from a starting point $\vct{x}$ under repeated applications
of a linear transformation $\matr{A}$. Answering two questions posed by
Kannan and Lipton in the 1980s, we show that when $\vsp{V}$ has dimension
one, this problem is solvable in polynomial time, and when $\vsp{V}$ has 
dimension two or three, the problem is in $\nprp$.

\end{abstract}

%% file: secIntroduction.tex
\section{Introduction}

The \emph{Orbit Problem} was introduced by Harrison in 
\cite{Har69} as a formulation of the reachability problem
for linear sequential machines. The problem is stated as follows:
\begin{quotation}\noindent
Given a square matrix $\matr{A}\in\rat^{m\times m}$ and vectors
$\vct{x},\vct{y}\in\rat^m$, decide whether there exists a non-negative
integer $n$ such that $\matr{A}^n\vct{x}=\vct{y}$.
\end{quotation}

The decidability of this problem remained open for over ten years, 
until it was shown to be decidable in polynomial time by Kannan and
Lipton \cite{KL80}. In the conclusion of the journal version of 
their work \cite{orbit}, the authors discuss a higher-dimensional
extension of the Orbit Problem, as follows:
\begin{quotation}\noindent
Given a square matrix $\matr{A}\in\rat^{m\times m}$, a vector $\vct{x}\in
\rat^m$, and a subspace $\vsp{V}$ of $\rat^m$, decide whether
there exists a non-negative integer $n$ such that $\matr{A}^n\vct{x}\in\vsp{V}$.
\end{quotation}

As Kannan and Lipton point out, the higher-dimensional Orbit Problem
is closely related to the \emph{Skolem Problem}: given a square matrix
$\matr{A}\in\rat^{m\times m}$ and vectors $\vct{x}, \vct{y}\in\rat^m$,
decide whether there exists a non-negative integer $n$ such that
$\vct{y^T}\matr{A}^n\vct{x}=0$.  Indeed, the Skolem Problem is the
special case of the higher-dimensional Orbit Problem in which the
target space $\vsp{V}$ has dimension $m-1$. The sequence of numbers
$\lrs{u}$ given by $u_n = \vct{y^T}\matr{A}^n\vct{x}$ is a linear
recurrence sequence. A well-known result, the Skolem-Mahler-Lech
Theorem, states that the set $\{n:u_n=0\}$ of zeros of any linear
recurrence is the union of a finite set and finitely many arithmetic
progressions \cite{mahlerSMLthm,lechSMLthm,skolemSMLthm,hanselSMLthm}.
Moreover, it is known how to compute effectively the arithmetic
progressions in question \cite{BerstelMignotte76}. The main difficulty
in deciding the Skolem Problem is thus to determine whether
the finite component of the set of zeros is empty.

The decidability of the Skolem Problem has been open for many decades
\cite{TUCS05,Tao08}, and it is therefore unsurprising that there has
been virtually no progress on the higher-dimensional Orbit Problem
since its introduction in \cite{orbit}. In fact, decidability of the
Skolem Problem for matrices of dimension three and four
\cite{mignotte,vereshchagin} was only established slightly prior to
the publication of \cite{orbit}, and there has been no substantial
progress on this front since.\footnote{A proof of decidability of the
  Skolem Problem for linear recurrence sequences of order five was
  announced in \cite{TUCS05}.  However, as pointed out in \cite{rp},
  the proof seems to have a serious gap.} In terms of lower bounds,
the strongest known result for the Skolem Problem is $\np$-hardness
\cite{blondel}, which therefore carries over to the unrestricted
version of the higher-dimensional Orbit Problem.

Kannan and Lipton speculated in \cite{orbit} that for target spaces of
dimension one the Orbit Problem might be solvable, ``hopefully with a 
polynomial-time bound''. They moreover observed that the cases in which
the target space $\vsp{V}$ has dimension two or three seem ``harder'', and proposed
this line of research as an approach towards the Skolem Problem. In spite
of this, to the best of our knowledge, no progress has been recorded on the
higher-order Orbit Problem in the intervening two-and-a-half decades.

Our main result is the following. We show that the higher-dimensional
Orbit Problem is in $\ptime$ if the target space has dimension one,
and in $\nprp$ if the target space has dimension two or three.  While
we make extensive use of the techniques of
\cite{mignotte,vereshchagin} on the Skolem Problem, our results, in
contrast, are independent of the dimension of the matrix
$\matr{A}$. 

Strictly speaking, Kannan and Lipton's original work on the Orbit
Problem concerned the case that the target was an affine subspace of
dimension $0$.  Our main results entail an $\nprp$ complexity bound in
case the target is an affine subspace of dimension $1$ or $2$, simply
by embedding into a vector-space problem one dimension higher.

The following example illustrates some of the phenomena that emerge in the
Orbit Problem for two-dimensional target spaces. Consider the following
matrix and initial vector:
\[
\matr{A} = 
\left[\begin{array}{rrrrrrr}
4 & & 6 && 14 && 21 \\
-8 & & -2 && -28 && -7 \\
-2 & & -3 && -6 && -9 \\
4 & & 1 && 12 && 3 \\
\end{array}\right]
\;\;\;
\vct{x} = \left[\begin{array}{r}
28 \\ -14 \\ -10 \\ 5
\end{array}\right]
\]
Then with target space
\[ \vsp{V} = \{(u_1,u_2,u_3,u_4)\in\rat^4:4u_1+7u_3=0,4u_2+7u_4=0\} \]
it can be shown that $\matr{A}^n\vct{x}\in\vsp{V}$ if and only if $n$ has residue
2 modulo 6. Such periodic behaviour can be analysed in terms of the 
eigenvalues of the matrix $\matr{A}$. These are $\lambda\omega$, $\overline{\lambda}
\omega$, $\lambda\overline{\omega}$ and $\overline{\lambda\omega}$, where
$\omega=e^{\pi i/3}$ is a primitive 6-th root of unity and $\lambda =
(-1+i\sqrt{39})/2$. The key observation is that the eigenvalues of $\matr{A}$ 
fall into only two classes under the equivalence relation $\sim$, defined
by $\alpha\sim\beta$ if and only if $\alpha/\beta$ is a root of unity. 

We handle such instances by analysing the equivalence classes of $\sim$.  We
show that, provided $\sim$ has sufficiently many equivalence classes, there is
at most one exponent $n$ such that $\matr{A}^n\vct{x}\in\vsp{V}$. Computable bounds on such an
$n$ are obtained utilising the work of \cite{mignotte,vereshchagin},
quantifying and strengthening some of the bounds given for the Skolem Problem. In
the case of a one-dimensional target subspace $\vsp{V}$, the resulting bound is
polynomial in the size of the problem representation, allowing for all
exponents $n$ up to the bound to be checked directly and yielding a
polynomial-time algorithm.  Unfortunately, when $\vsp{V}$ has dimension two or three,
the bounds on $n$ are exponential in the size of the input, leading to an
$\nprp$ \emph{guess-and-check} procedure, in which an
$\rp$ oracle is used to check whether $\matr{A}^n\vct{x}\in\vsp{V}$ for a guessed value
of $n$. Finally, the case in which the eigenvalues of $\matr{A}$ have fewer
equivalence classes under $\sim$ is handled explicitly using a case analysis on
the residue of $n$ modulo the least common multiple of the orders of all ratios
of eigenvalues which are roots of unity. For each such residue class, we show
how to determine whether it contains exponents $n$ for which $\matr{A}^n\vct{x}\in\vsp{V}$.
Noting that there are at most exponentially many such residue classes, we can
directly incorporate this case analysis into an $\nprp$ algorithm.

\subsection{Related Work}

Aside from its connection to the Skolem Problem, the
higher-dimensional Orbit Problem is closely related to termination
problems for linear programs (see, e.g., \cite{benamram,Bra06,tiwari})
and to reachability questions for discrete linear dynamical systems
(cf. \cite{TUCS05}). Another related problem is the \emph{Polyhedron
  Hitting Problem}, which replaces the target space with an
intersection of half-spaces. In \cite{TV11}, the Polyhedron Hitting
Problem is related to decision problems in formal language theory.
Some partial decidability results for this problem are given
in~\cite{COWsodaPHP}. Let us also mention the more recent work of
Arvind and Vijayaraghavan \cite{gapl} which places the original Orbit
Problem in the logspace counting hierarchy $\cclass{GapLH}$.

Another generalisation of the Orbit Problem was considered in \cite{cai} and
shown to be decidable in polynomial time. This asks, given commuting rational
matrices $\matr{A}$, $\matr{B}$, and $\matr{C}$, whether there exist integers 
$i$ and $j$ such that $\matr{A}^i\matr{B}^j=\matr{C}$.

A continuous version of the Orbit Problem is considered in \cite{hainry}. Here
one studies linear differential equations of the form $\vct{x'}(t)=\matr{A}\vct{x}(t)$ for a
rational matrix $\matr{A}$. The problem is to decide, for a given initial condition
$\vct{x}(0)$ and target vector $\vct{v}$, whether there exists $t$ such that $\vct{x}(t)=\vct{v}$. 
The main result of \cite{hainry} shows decidability of this problem.

%% file: orbit.secOutline.tex
\section{Paper outline}\label{sec: orbit.outline}

This work is based on our conference paper~\cite{COW13}.
The main technical results are the following theorems:
\begin{theorem}\label{thm: orbit.bounds}
Suppose we are given an instance of the Orbit Problem, comprising a square 
matrix $\matr{A}\in\rat^{m\times m}$, a vector $\vct{x}\in\rat^m$ and a subspace
$\vsp{V}\subseteq\rat^m$ with $\dim(\vsp{V})\leq 3$. Let $\len{I}$ be the length of the description 
of the input data. There exists a bound $N=\ein{I}$ such 
that if the instance is positive, then there exists a witness (that is, $n\in\nat$ with 
$\matr{A}^n\vct{x}\in\vsp{V}$) such that $n<N$.
\end{theorem}

\begin{theorem}\label{thm: orbit.complexity}
The Orbit Problem with $\dim(\vsp{V})\leq3$ is in $\nprp$. Further, if 
$\dim(\vsp{V})=1$, then the problem is in $\ptime$.
\end{theorem}
In this section we give a high-level overview of the argument.

Firstly, we must emphasise that the fixed-dimensional versions 
of the Orbit Problem referred to by Theorems~\ref{thm: orbit.bounds} and~\ref{thm: orbit.complexity}
are closely related to the  Skolem Problem for linear recurrence sequences of order at most four. 
In the interest of clarity, we have confined our treatment of the  Skolem Problem 
to the Appendix. We employ two powerful tools from transcendence theory, due to 
Baker-W{\"u}stholz and van der Poorten, as well as standard results from algebraic 
number theory, to prove our main result on the  Skolem Problem, 
Theorem~\ref{thm: skolem.bounds}, which shows the existence of effective bounds on 
the zeros of LRS of order at most four and upon which our results on the Orbit Problem
build.

We now outline the structure of the argument which establishes Theorems~\ref{thm: orbit.bounds}
and~\ref{thm: orbit.complexity}. The first step is a reduction to a similar problem, a 
polynomial version of the \emph{matrix power problem:} 
given a rational square matrix $\matr{A}$ and polynomials
$P_1,\dots,P_d\in\rat[x]$ such that $P_1(\matr{A}),\dots,P_d(\matr{A})$ 
are linearly independent over $\rat$, determine whether there exists $n$ such that 
$\matr{A}^n$ lies in the $\rat$-vector space 
$\spa\{P_1(\matr{A}),\dots,P_d(\matr{A})\}$. The reduction does not
increase the dimension of
the target space, so we will always have $d\leq 3$. The reduction 
can be carried out in polynomial time
and rests entirely on standard techniques from linear algebra.

For the second step, we construct a \emph{Master System.} This is a system of
equations, based on the eigenvalues of $\matr{A}$ and the polynomials $P_1,\dots,P_d$.
It has $d+1$ unknowns: the exponent $n$ and the coefficients 
$\orbitcoeff_1,\dots,\orbitcoeff_d$ which witness the membership of $\matr{A}^n$
in $\spa\{P_1(\matr{A}),\dots,P_d(\matr{A}))\}$. The solutions 
$(n,\orbitcoeff_1,\dots,\orbitcoeff_d)$ of the Master System will be exactly the solutions
of the matrix equation 
$\matr{A}^n=\orbitcoeff_1P_1(\matr{A})+\dots+\orbitcoeff_dP_d(\matr{A})$.
The domain of $n$ is $\nat$ throughout. Since the input data is
rational, any solution $(n,\orbitcoeff_1,\dots,\orbitcoeff_d)$ of the Master System
will necessarily have $\orbitcoeff_1,\dots,\orbitcoeff_d\in\rat$.

Next, in Section \ref{secOnedim}, we give a polynomial-time decision procedure to 
determine whether the Master System for an instance with a one-dimensional
target space has a solution. The algorithm explicitly manipulates the equations
in the system, preserving the set of solutions at every step, to determine the existence
of a solution in polynomial time, settling the one-dimensional case of Theorem
\ref{thm: orbit.complexity}. 
The section rests critically on Theorem \ref{thm: skolem.bounds}
for non-degenerate linear recurrence sequences of order $2$,
which allows us to bound the exponent
in all cases when $\matr{A}$ has two eigenvalues whose ratio is not a root of unity.
In all other situations, the given Orbit instance essentially reduces to a 
system of linear congruences, easily solved using the Chinese Remainder Theorem.
The solution method yields the full set of witness exponents $n$ when
this set is finite, or a description of the witness set as an arithmetic progression when 
it is infinite. Thus, if the problem instance is positive, a witness 
exponent which is at most exponentially large is automatically guaranteed to 
exist, as promised by Theorem \ref{thm: orbit.bounds}, by virtue of our ability 
to write it down using polynomially many bits.

An important concept for the cases of two- and three-dimensional target spaces
is the notion of \emph{degeneracy}.
An instance $(\matr{A},\vct{x},\vsp{V})$ of the Orbit Problem is defined as 
degenerate if there exist two distinct eigenvalues of $\matr{A}$ whose
quotient is a root of unity, otherwise the instance is non-degenerate.
In general, it is possible to reduce an arbitrary Orbit Problem instance to a set
of non-degenerate instances.
Let $L$ be the least common multiple of the orders of all quotients of eigenvalues
of $\matr{A}$ which are roots of unity. For each $j\in\{ 0,\dots,L-1\}$, consider
separately the problem of deciding whether there exists $n\in\nat$ such
that $(\matr{A}^L)^{n}\left(\matr{A}^j\vct{x}\right)\in\vsp{V}$. 
These instances are all non-degenerate,\footnote{Indeed, the eigenvalues of $\matr{A}^L$ are 
exactly $\lambda_i^L$ where $\lambda_i$ are the eigenvalues of
$\matr{A}$. If for any two distinct such eigenvalues, say $\lambda_i^L\neq
\lambda_j^L$, we have $\left(\lambda_i^L/ \lambda_j^L\right)^t=1$, then
$\lambda_i/\lambda_j$ must also be a root of unity. Then by the definition of
$L$, $\lambda_i^L/\lambda_j^L=1$, which gives the contradiction
$\lambda_i^L=\lambda_j^L$.} and the original problem instance is positive if and 
only if at least one of these $L$ non-degenerate instances is positive. 
Unfortunately, this reduction to the non-degenerate case carries an exponential
overhead, as $L$ is, in general, exponentially large in the size of the input data.

Instead, we adopt the following strategy for solving the Orbit Problem 
for possibly degenerate instances.  
Assume that as part of the input, we are given the residue $r=n\bmod L$. 
Thus, we are interested in determining whether the Master System has a solution
$(n,\orbitcoeff_1,\dots,\orbitcoeff_d)$ with exponent $n$ such that $r=n\bmod L$.
We will prove that for any $r$, there exists a bound $N_r$ such that 
if there exists such an exponent with residue $r$, then one exists which does
not exceed the bound $N_r$. Furthermore, $N_r=\ein{I'}$, where 
$\len{I'}=\len{I}+\len{r}$ is the length of the input augmented with 
the binary representation of $r$. This is clearly sufficient to prove 
Theorem \ref{thm: orbit.bounds}: simply take $N=\max\{N_r:r\in\{0,\dots,L-1\}\}$.
The case analysis on $r$ simplifies the Master System considerably,
effectively eliminating degeneracy as a concern, and allowing us to derive
the existence of $N_r$ using our results on the  Skolem Problem
for LRS of order $3$ and $4$.
For each fixed $r$, algebraic manipulation yields either a `small' witness 
$n$ of the correct residue, or a non-degenerate linear recurrence sequence $\lrs{u}$
of low order
such that if the Master System has a solution with exponent $n$ with the desired
residue $r$, then $u_n=0$. The description of this linear
recurrence sequence is computable in polynomial time from the input instance
and $r$. Since $\len{r}=\pin{I}$, it follows that the length of the description
of $\lrs{u}$ is $\len{u}=\pin{I'}=\pin{I}$, so by Theorem \ref{thm: skolem.bounds}, 
the desired bound $N_r$ exists and $N_r=\ein{I}$.

We must emphasise that this algebraic manipulation of the Master
System and the calculation of the description of $\lrs{u}$ is not part
of the decision procedure for the Orbit Problem. Rather its purpose is
to prove the existence of the desired bounds $N_r$, and hence of
$N$. We make use of the observation that this manipulation can, in
principle, be carried out in polynomial time, to conclude
$N_r=\ein{I'}$ and $N=\ein{I}$, and hence establish Theorem \ref{thm:
  orbit.bounds}.

Given the bound $N$ of Theorem \ref{thm: orbit.bounds}, we employ
a \emph{guess-and-check} procedure to obtain the complexity upper
bounds of Theorem \ref{thm: orbit.complexity}.
Since $N$ is at most exponentially large in the size of the input, 
an $\np$ procedure can guess an exponent $n$ such that $n<N$.
Then we compute $\matr{A}^n\vct{x}$
by iterated squaring, thereby using polynomially many arithmetic operations.
Moreover, all integers that occur in this algorithm have a polynomial-sized
representation via arithmetic circuits. 
Now, to verify $\matr{A}^n\vct{x}\in\vsp{V}$, we compute the determinant
of $\matr{B^T}\matr{B}$, where $\matr{B}$ is the matrix whose columns are 
$\matr{A}^n\vct{x}$ and the basis vectors specifying $\vsp{V}$, also
as an arithmetic circuit.
Clearly, $n$ is a witness to the problem instance if and only if this determinant is zero.
This is easy to determine with an $\eqslp$ oracle, so we have membership in 
$\npeqslp$. It is known that $\eqslp\subseteq\corp$ \cite{schonhage}, so we
have membership in $\nprp$, thereby establishing Theorem \ref{thm: orbit.complexity}.

%% file: orbit.secReduction.tex
\section{Reduction}\label{secReduction}
\subsection{Matrix power problem}

Suppose we are given a matrix $\matr{A}\in\rat^{m\times m}$, a vector $\vct{x}\in\rat^m$ and a target
vector space $\vsp{V}\subseteq \rat^m$ specified by a basis of rational vectors $\vct{y_1},\dots,\vct{y_k}$. We
wish to decide whether there exists $n\in\nat$ such that $\matr{A}^n\vct{x}\in \vsp{V}$. 

Observe that we can rescale $\matr{A}$ in polynomial time by the least common multiple
of all denominators appearing in $\matr{A}$. This reduces the general problem to the
sub-problem in which $\matr{A}$ is an integer matrix.
 
Let
$\nu=\max\{ m\:|\: \vct{x},\matr{A}\vct{x},\dots,\matr{A}^m\vct{x}
\mbox{ are linearly independent}\}$, 
$B=\{ \vct{x},\matr{A}\vct{x},\dots,\matr{A}^\nu \vct{x}\}$, 
$\vsp{U}=\spa(B)$ and
$\matr{D}=\left[\begin{array}{cccc} \vct{x} & \matr{A}\vct{x} & \dots 
& \matr{A}^\nu \vct{x}\end{array}\right]$.  It
is clear that $\vsp{U}$ is invariant under the linear transformation $\matr{A}$, so consider
the restriction of $\matr{A}$ to $\vsp{U}$. Suppose $\vct{b}=(b_0,\dots, b_\nu)^T$ are the
coordinates of $\matr{A}^{\nu+1}\vct{x}$ with respect to ${B}$, that is,
$\matr{A}^{\nu+1}\vct{x}=\matr{D}\vct{b}$. 
The restriction of $\matr{A}$ to $\vsp{U}$ with respect to the basis $B$ is described by the matrix 
\[
\matr{M}=\left[\begin{array}{ccccc}
0 & 0 & \dots & 0 & b_0\\
1 & 0 & \dots & 0 & b_1\\
0 & 1 & \dots & 0 & b_2\\
0 & 0 & \ddots & 0 & \vdots\\
0 & 0 & \dots & 1 & b_\nu
\end{array}\right].
\]
It is easy to check that $\matr{D}\matr{M}=\matr{A}\matr{D}$. 
Thus, if some vector $\vct{z}$ has coordinates
$\vct{z'}$ with respect to $B$, so that $\vct{z}=\matr{D}\vct{z'}$, 
then $\matr{A}\vct{z}$ has coordinates $\matr{M}\vct{z'}$
with respect to $B$, so that $\matr{A}\vct{z}=\matr{D}\matr{M}\vct{z'}$. By induction, for all
$n\in\nat$, $\matr{A}^n\vct{x}=\matr{D}\matr{M}^n\vct{x'}$, 
where $\vct{x'}=(1,0,\dots,0)^T$. Next we calculate
a basis $\{\vct{w_1},\dots,\vct{w_t}\}$
for $\vsp{W}\defn\vsp{U}\cap \vsp{V}$.  Then, if $\vct{w_i}$ are such that
$\vct{w_i}=\matr{D}\vct{w_i'}$ for all $i$, we have 
\[ \matr{A}^n\vct{x}\in \vsp{V} \iff \matr{A}^n\vct{x}\in \vsp{W} 
\iff \matr{M}^n\vct{x'}\in \spa\{\vct{w_1'},\dots,\vct{w_t'}\}. \]
Notice that the matrix $\matr{M}$ describes a restriction of the linear transformation denoted by
$\matr{A}$, so its eigenvalues are a subset of the eigenvalues of $\matr{A}$. In particular,
since $\matr{A}$ was rescaled to an integer matrix, the eigenvalues of $\matr{M}$ are
algebraic integers as well.

Define the matrices $\matr{T_1},\dots,\matr{T_t}$ by 
\[
\matr{T_i}=\left[\begin{array}{cccc}
\vct{w_i'} & \matr{M}\vct{w_i'} & \dots & \matr{M}^\nu \vct{w_i'}\end{array}\right].
\]
We will show that $\matr{M}^n\vct{x'}\in \spa\{\vct{w_1'},\dots,\vct{w_t'}\}$ 
if and only if $\matr{M}^n\in
\spa\{\matr{T_1}, \dots,\matr{T_t}\}$. If for some coefficients $\orbitcoeff_i$ we have 
\[
\matr{M}^n=\sum_{i=0}^t\orbitcoeff_i\matr{T_i},
\]
then considering the first column of both sides, we have 
\[
\matr{M}^n\vct{x'}=\sum_{i=0}^t\orbitcoeff_i\vct{w_i'}.
\]
Conversely, suppose $\matr{M}^n\vct{x'}=\sum_{i=0}^t\orbitcoeff_i\vct{w_i'}$. 
Then note that $\vct{x'},\matr{M}\vct{x'},
\dots,\matr{M}^\nu\vct{x'}$ are just the unit vectors of size $\nu+1$. Multiplying by $\matr{M}^j$
for $j=0,\dots,\nu$ gives $\matr{M}^{n+j}\vct{x'}=\sum_{i=0}^t\orbitcoeff_i\matr{M}^j\vct{w_i'}$. 
The left-hand side is exactly the $(j+1)$-th column of $\matr{M}^n$, whereas 
$\matr{M}^j\vct{w_i'}$ on the
right-hand side is exactly the $(j+1)$-th column of $\matr{T_i}$. So we have
$\matr{M}^n=\sum_{i=0}^t\orbitcoeff_i\matr{T_i}$.

Thus, we have reduced the Orbit Problem to the \emph{Matrix Power
  Problem:} determining whether some power of a given matrix lies
inside a given vector space of matrices. Now we will perform a further
reduction step. It is clear that within the space
$\vsp{T}\defn\spa\left\{ \vct{T_1},\dots,\vct{T_t}\right\}$ it
suffices to consider only matrices of the shape $P(\matr{M})$ where
$P\in\rat[x]$. We find a basis for the space
$\vsp{P}\defn\left\{ P(\matr{M})\:|\: P\in\rats[x]\right\} $ and then
a basis $\left\{ P_1(\matr{M}),\dots,P_s(\matr{M})\right\}$ for
$\vsp{P}\cap \vsp{T}$.  Then
$\matr{M}^n\in \vsp{T}\iff \matr{M}^n\in \vsp{P}\cap \vsp{T}$. We call
the problem of determining, given $\matr{M}$ and $P_1,\dots,P_s$,
whether there exists $n\in\nat$ such that
$\matr{M}^n\in\spa\{P_1(\matr{M}),\dots,P_s(\matr{M})\}$, the
\emph{polynomial version} of the matrix power problem. Observe that
$\dim(\vsp{V})\geq\dim(\vsp{T})\geq\dim(\vsp{T}\cap \vsp{P})$, so the
dimension of the target vector space does not grow during the
described reductions.  All described operations may be performed in
polynomial time using standard techniques from linear algebra.

\subsection{Master System of equations}

Suppose now we have an instance $(\matr{A},P_1,\dots,P_s)$ of the polynomial version
of the matrix power problem. Calculate the minimal polynomial of $\matr{A}$
and obtain canonical representations of its roots $\alpha_1,\dots,\alpha_k$,
that is, the eigenvalues of $\matr{A}$. This may be done in polynomial time, see
Section \ref{appAlgnum}. Throughout, for an eigenvalue $\alpha_i$
we will denote by $\mul(\alpha_i)$ the multiplicity of $\alpha_i$ as a
root of the minimal polynomial of the matrix.

Fix an exponent $n\in\nat$ and coefficients $\orbitcoeff_1,\dots,\orbitcoeff_s\in\cplx$
and define the polynomials $P(x)=\sum_{i=1}^s\orbitcoeff_iP_i(x)$ and $Q(x)=x^n$. It is
easy to see that
\[ Q(\matr{A})=P(\matr{A}) \]
if and only if
\begin{equation}
\forall i\in\{1,\dots,k\}.
\forall j\in\{0,\dots,\mul(\alpha_i)-1\}.
P^{(j)}(\alpha_i)=Q^{(j)}(\alpha_i).
\label{eq: system of equations}
\end{equation}
Indeed, $P-Q$ is zero at $\matr{A}$ if and only if the minimal polynomial of $\matr{A}$
divides $P-Q$, that is, each $\alpha_i$ is a root of $P-Q$ with multiplicity at least
$\mul(\alpha_i)$, or equivalently, each $\alpha_i$ is a root of $P-Q$ and its 
first $\mul(\alpha_i)-1$ derivatives.

Thus, in order to decide whether there exists an exponent $n$ and coefficients
$\orbitcoeff_i$ such that $\matr{A}^n=\sum_{i=1}^s\orbitcoeff_iP_i(\matr{A})$, 
it is sufficient to solve the system
of equations (\ref{eq: system of equations}) where the unknowns are
$n\in\nat$ and $\orbitcoeff_1,\dots,\orbitcoeff_s\in\cplx$.  Each eigenvalue $\alpha_i$
contributes $\mul(\alpha_i)$ equations which specify that $P(x)$ and
its first $\mul(\alpha_i)-1$ derivatives all vanish at $\alpha_i$. 

For brevity in what follows, we will denote by $\eq(\alpha_i,j)$ the
$j$-th derivative equation contributed to the system by $\alpha_i$, that is,
$P^{(j)}(\alpha_i)=Q^{(j)}(\alpha_i)$.  This notation is defined only for
$0\leq j<\mul(\alpha_i)$.  We will also denote by
$\Eq(\alpha_i)$ the set of equations contributed by $\alpha_i$ to the
system: 
\[
\Eq(\alpha_i)=\left\{ \eq(\alpha_i,0),\dots,
\eq(\alpha_i,\mul(\alpha_i)-1)\right\}.
\]

For example, if the minimal polynomial of $\matr{A}$ has roots 
$\alpha_1$, $\alpha_2$, $\alpha_3$ with
multiplicities $\mul(\alpha_i)=i$ and the target space is $\spa\left\{
P_1(\matr{A}),P_2(\matr{A})\right\}$ then the system contains six equations:
\begin{align*}
\alpha_1^n & = \orbitcoeff_1P_1(\alpha_1)+\orbitcoeff_2P_2(\alpha_1)\\
\alpha_2^n & = \orbitcoeff_1P_1(\alpha_2)+\orbitcoeff_2P_2(\alpha_2)\\
n\alpha_2^{n-1} & = \orbitcoeff_1P_1'(\alpha_2)+\orbitcoeff_2P_2'(\alpha_2)\\
\alpha_3^n & = \orbitcoeff_1P_1(\alpha_3)+\orbitcoeff_2P_2(\alpha_3)\\
n\alpha_3^{n-1} & = \orbitcoeff_1P_1'(\alpha_3)+\orbitcoeff_2P_2'(\alpha_3)\\
n(n-1)\alpha_3^{n-2} & =\orbitcoeff_1P_1''(\alpha_3)+\orbitcoeff_2P_2''(\alpha_3)
\end{align*}
Then $\eq(\alpha_3,0)$ is the equation 
\[
\alpha_3^n=\orbitcoeff_1P_1(\alpha_3)+\orbitcoeff_2P_2(\alpha_3)
\]
and $\Eq(\alpha_2)$ is the two equations
\begin{align*}
\alpha_2^n & = \orbitcoeff_1P_1(\alpha_2)+\orbitcoeff_2P_2(\alpha_2)\\
n\alpha_2^{n-1} & = \orbitcoeff_1P_1'(\alpha_2)+\orbitcoeff_2P_2'(\alpha_2)
\end{align*}

%% file: orbit.secDim1.tex
\section{One-dimensional target space}
\label{secOnedim}

Suppose we are given a one-dimensional matrix power problem instance $(\matr{A},P)$
and wish to decide whether $\matr{A}^n\in \spa\{P(\matr{A})\}$ for some $n$. We have
constructed a system of equations in the exponent $n$ and the coefficient $\orbitcoeff$
as in (\ref{eq: system of equations}). For example, if the roots of 
the minimal polynomial of $\matr{A}$
are $\alpha_1,\alpha_2,\alpha_3$ with multiplicities $\mul(\alpha_j)=j$, 
the system is:
\begin{align*}
\alpha_1^n & =  \orbitcoeff P(\alpha_1)\\
\alpha_2^n & =  \orbitcoeff P(\alpha_2)\\
n\alpha_2^{n-1} & =  \orbitcoeff P'(\alpha_2)\\
\alpha_3^n & =  \orbitcoeff P(\alpha_3)\\
n\alpha_3^{n-1} & =  \orbitcoeff P'(\alpha_3)\\
n(n-1)\alpha_3^{n-2} & =  \orbitcoeff P''(\alpha_3)
\end{align*}
In this section we will describe how such systems may be solved in polynomial
time. First, we perform some preliminary calculations.

\begin{enumerate}
\item We check whether $\orbitcoeff=0$ has a corresponding $n$ which solves the matrix
equation $\matr{A}^n=\orbitcoeff P(\matr{A})$, that is, whether $\matr{A}$ is nilpotent.
Otherwise, assume $\orbitcoeff\neq0$.
\item Let $k=\max_j\{\mul(\alpha_j)\}$.  We check for all $n<k$ whether
$\matr{A}^n$ is a multiple of $P(\matr{A})$. If so, we are done. Otherwise, assume $n\geq
k$.
\item We check whether $\alpha_j=0$ for some $j$. If so, then all of the
equations $\Eq(\alpha_i)$ are of the form $0=\orbitcoeff P^{(t)}(0)$, which is
equivalent to $0=P^{(t)}(0)$.  We can easily check whether these equations are
satisfied. If so, we dismiss them from the system without changing the set of
solutions.  If not, then there is no solution and we are done. Now we assume
$\alpha_j\neq0$ for all $j$.
\item Finally, we check whether the right-hand side $\orbitcoeff P^{(t)}(\alpha_j)$ of
some equation is equal to $0$, by dividing $P^{(t)}(x)$ by the minimal polynomial of $\alpha_j$.
If this is the case, then the problem instance is negative, because the left-hand 
sides are all non-zero.
\end{enumerate}
Let $\eq(\alpha_i,k)/\eq(\alpha_j,t)$ denote the equation obtained from 
$\eq(\alpha_i,k)$ and $\eq(\alpha_j,t)$ by asserting that the ratio
of the left-hand sides equals the ratio of the right-hand sides, that is,
\[
\frac{n(n-1)\dots(n-k+1)\alpha_i^{n-k}}{n(n-1)\dots(n-t+1)\alpha_j^{n-t}}=
\frac{P^{(k)}(\alpha_i)}{P^{(t)}(\alpha_j)}.
\]
We compute representations of all quotients $\alpha_i/\alpha_j$, and consider
three cases.

\emph{Case I.} Some quotient $\alpha_i/\alpha_j$ is not a root
of unity. Then $\eq(\alpha_i,0)$ and $\eq(\alpha_j,0)$
together imply $\eq(\alpha_i,0)/\eq(\alpha_j,0)$,
that is,
\[ \left(\frac{\alpha_i}{\alpha_j}\right)^n=\frac{P(\alpha_i)}{P(\alpha_j)}. \]

In Section \ref{appAlgnum}, we discuss the efficient representation and
manipulation of algebraic numbers. By Lemma \ref{lem:operations on algebraic
numbers}, we can compute representations of $P(\alpha_i)/P(\alpha_j)$ and
$\alpha_i/\alpha_j$ in polynomial time. Then by Lemma \ref{lem: algebraic
number power problem} in Section \ref{appSkolem2}, $n$ is bounded by a
polynomial in the input. We check $\matr{A}^n\in \spa\{P(\matr{A})\}$ for 
all $n$ up to the bound and we are done.

\emph{Case II.} All quotients $\alpha_i/\alpha_j$ are roots of unity, and all
roots of the minimal polynomial of $\matr{A}$ are simple. Then the system is equivalent to 
\[
\orbitcoeff=\frac{\alpha_1^n}{P(\alpha_1)}\wedge\bigwedge_{i<j}
\frac{\eq(\alpha_i,0)}{\eq(\alpha_j,0)}.
\]
It is sufficient to determine whether there exists some $n$ which satisfies
\begin{equation}\label{eq: 1d orbit, all ratio equations}
\bigwedge_{i<j}\frac{\eq(\alpha_i,0)}{\eq(\alpha_j,0)}.
\end{equation}
Consider each equation $\eq(\alpha_i,0)/\eq(\alpha_j,0)$:
\begin{equation}\label{eq: 1d orbit, one ratio equation}
\left(\frac{\alpha_i}{\alpha_j}\right)^n=\frac{P(\alpha_i)}{P(\alpha_j)}.
\end{equation}
Suppose $\alpha_i/\alpha_j$ is an $r$-th root of unity. If the right-hand side
of (\ref{eq: 1d orbit, one ratio equation}) is also an $r$-th root of unity,
then the solutions of (\ref{eq: 1d orbit, one ratio equation}) are $n\equiv
t\bmod r$ for some $t$. If not, then (\ref{eq: 1d orbit, one ratio
equation}) has no solution, so the entire system (\ref{eq: system of
equations}) has no solution, and the problem instance is negative. By Lemma
\ref{lem:operations on algebraic numbers}, we can determine in polynomial time
whether the right-hand side of (\ref{eq: 1d orbit, one ratio equation}) is a
root of unity, and if so, calculate $t$. We transform each equation in
(\ref{eq: 1d orbit, all ratio equations}) into an equivalent congruence in $n$.
This gives a system of congruences in $n$ which is equivalent to (\ref{eq: 1d
orbit, all ratio equations}).  We solve it using the Chinese Remainder Theorem.
The problem instance is positive if and only if the system of congruences has a
solution.

\emph{Case III.} All quotients $\alpha_i/\alpha_j$ are roots of unity, and
$f_A(x)$ has repeated roots. We transform the system into an equivalent one in
the following way. First, we include in the new system all the quotients of
equations $\eq(\alpha_i,0)$ as in Case 2. Second, for each repeated
root $\alpha_i$ of $f_A(x)$, we take the quotients
$\bigwedge_{j=0}^{\mul(\alpha_i)-2}\eq(\alpha_i,j)/\eq(\alpha_i,j+1)$.
Third, we include the equation $\orbitcoeff=\alpha_1/P(\alpha_1)$.
\[
\bigwedge_{i<j}\frac{\eq(\alpha_i,0)}{\eq(\alpha_j,0)}
\wedge\bigwedge_i\bigwedge_{j=0}^{\mul(\alpha_i)-2}
\frac{\eq(\alpha_i,j)}{\eq(\alpha_i,j+1)}
\wedge \orbitcoeff=\frac{\alpha_1}{P(\alpha_1)}.
\]
We solve the first conjunct as in Case 2. If there is no solution, then we are
done. Otherwise, the solution is some congruence $n\equiv t_1\mbox{ mod }t_2$.
For the remainder of the system, each ratio
$\eq(\alpha_i,j)/\eq(\alpha_i,j+1)$ contributed by a repeated
root $\alpha_i$ has the shape 
\[ \frac{\alpha_i}{n-j}=\frac{P^{(j)}(\alpha_i)}{P^{(j+1)}(\alpha_i)}, \]
which is equivalent to
\begin{equation}\label{eq: 1d orbit, rep root eqn}
n=j+\frac{P^{(j+1)}(\alpha_i)}{P^{(j)}(\alpha_i)}\alpha_i.
\end{equation}
For each such equation (\ref{eq: 1d orbit, rep root eqn}), we calculate the
right-hand side in polynomial time, using the methods outlined in Section
\ref{appAlgnum}, and check whether it is in $\nat$. If not, then the
system has no solution. Otherwise, (\ref{eq: 1d orbit, rep root eqn}) points to
a single candidate $n_0$. We do this for all equations where $n$ appears
outside the exponent. If they point to the same value of $n$, then the system
is equivalent to 
\begin{align*}
n & \equiv  t_1\bmod t_2\\
n & =  n_0\\
\orbitcoeff & =  \alpha_1^n/P(\alpha_1)
\end{align*}
We check whether $n_0$ satisfies the congruence and we are done.

%% file: orbit.secDim2.tex
\section{Two-dimensional target space}
\label{secTwodim}

Suppose we are given a rational square matrix $\matr{A}$ and polynomials 
$P_1,P_2$ with rational coefficients such that $P_1(\matr{A})$ and
$P_2(\matr{A})$ are linearly independent over $\rat$. 
We want to decide whether there exists $n\in\nat$ such that
$\matr{A}^n$ lies in the $\rat$-vector space $\spa\{P_1(\matr{A}),P_2(\matr{A})\}$.
We have derived a Master System of equations (\ref{eq: system of equations})
in the unknowns $(n,\orbitcoeff_1,\orbitcoeff_2)$ whose solutions
are precisely the solutions of the matrix equation 
$\matr{A}^n=\orbitcoeff_1 P_1(\matr{A}) + \orbitcoeff_2 P_2(\matr{A})$.

In this section, we will show that there exists a bound
$N$, exponentially large in the size of the input, 
such that if the problem instance is positive, then there exists
a witness exponent $n$ with $n<N$. This will be sufficient to show that the 
problem is in the complexity class $\nprp$, as outlined earlier. 

Notice that we may freely assume that the eigenvalues of $\matr{A}$ are non-zero. 
Indeed, if $0$ is an eigenvalue, then consider $\eq(0,0)$:
\[ 0 = \orbitcoeff_1P_1(0) + \orbitcoeff_2P_2(0). \]
If at least one of $P_1(0)$, $P_2(0)$ is non-zero, then we have a linear
dependence between $\orbitcoeff_1, \orbitcoeff_2$. Then we express one of
the coefficients $\orbitcoeff_1,\orbitcoeff_2$ in terms of the other, obtaining
a Master System of dimension $1$, and then the claim follows inductively.
Otherwise, if $P_1(0)=P_2(0)=0$, then
$\eq(0,0)$ is trivially satisfied for all $n,\orbitcoeff_1,\orbitcoeff_2$, so we 
remove $\eq(0,0)$ from the Master System without altering the set of solutions. 
We examine in this way all equations contributed by
$0$, either removing them from the system, or obtaining a lower-dimensional
system which then yields the required bound $N$ inductively.

As outlined in Section \ref{sec: orbit.outline}, we show the existence of the bound $N$
by performing a case analysis on $n\bmod L$, where
\[ 
L = \mbox{lcm}\{ \mbox{order}(\lambda_i/\lambda_j) : \lambda_i,\lambda_j\mbox{ eigenvalues of $\matr{A}$
and $\lambda_i/\lambda_j$ root of unity} \}.
\]
We will show that for any fixed value $r\in\{0,\dots,L-1\}$, there exists a bound $N_r$,
exponentially large in the size of the input, such that if the Master System has a solution
with exponent of residue $r$ modulo $L$, then it has a solution with exponent $n$ such that
$n<N_r$. To obtain the bounds $N_r$, we show how the Master System can be manipulated algebraically
in polynomial time to yield a non-degenerate linear recurrence sequence of order $3$
whose zeros are a superset of the exponents $n$ which solve the Master System.
This manipulation is a proof technique to show the existence of the bound
$N_r$, not a feature of the algorithm. The decision method is instead the guess-and-check
procedure explained in Section \ref{sec: orbit.outline}.

Thus, from here onwards, we assume we are given a fixed $r$, which increases the input
size only polynomially, and are interested solely in exponents $n$ with $n\bmod L=r$.
Since we admit degenerate problem instances, we need to consider the 
relation $\sim$ on the eigenvalues of $\matr{A}$, defined by 
\[ \alpha\sim\beta\mbox{ if and only if $\alpha/\beta$ is a root of unity}. \]
It is clear that $\sim$ is an equivalence relation. The equivalence classes
${\cls C}_1,\dots,{\cls C}_k$ of $\sim$ are of two kinds.  First, a class can
be its own image under complex conjugation: 
\[ {\cls C}_i=\{\overline{\alpha}\:|\:\alpha\in{\cls C}_i\} \]
Each such self-conjugate class $\{\alpha_1,\dots,\alpha_s\}$ has the form
$\{\alpha\omega_1,\dots,\alpha\omega_s\}$ where $\omega_i$ are roots of unity,
and $|\alpha_j|=\alpha\in \ra$.  Call this $\alpha$ the
\emph{representative} of the equivalence class ${\cls C}_i$. Second, if an equivalence
class is not self-conjugate, then its image under complex conjugation must be
another equivalence class of $\sim$.  Thus, the remaining equivalence classes
of $\sim$ are grouped into pairs $({\cls C}_i,{\cls C}_j)$ such that ${\cls
C}_i= \{\overline{x}\:|\:x\in{\cls C}_j\} =\overline{{\cls C}_j}$.  In this
case, we can write ${\cls C}_i$ and ${\cls C}_j$ as
\[ {\cls C}_i=\{\lambda\omega_1,\dots,\lambda\omega_s\} \]
\[{\cls C}_j=\{\overline{\lambda\omega_1},\dots,\overline{\lambda\omega_s}\}\]
where $\omega_i$ are roots of unity, $\lambda\in\alg$ and $\arg(\lambda)$
is an irrational multiple of $2\pi$. Call $\lambda$ the representative of ${\cls
C}_i$ and $\overline{\lambda}$ the representative of ${\cls C}_j$.

Observe that the representatives
of self-conjugate classes are distinct positive real numbers, and that 
no ratio of representatives can be a root of unity.
Recall also that we can assume the eigenvalues of $\matr{A}$ are
algebraic integers, as a by-product of the reduction from the Orbit Problem.
Since roots of unity and their multiplicative inverses are algebraic integers,
it follows that the representatives of equivalence classes must also be algebraic
integers.

Let 
\[ \Eq({\cls C})=\bigcup_{\alpha\in{\cls C}}\Eq(\alpha) \]
denote the set of equations contributed to the system by the eigenvalues in
${\cls C}$, and let 
\[
Eq({\cls C},i)=\bigcup_{
\begin{array}{c}
\alpha\in{\cls C}\\
\mul(\alpha)>i
\end{array}}
\{\eq(\alpha,i)\} 
\]
denote the set of $i$-th derivative equations contributed by the roots in
${\cls C}$. 

To show the existence of the required bound $N_r$, we will perform a case analysis
on the number of equivalence classes of $\sim$.

\emph{Case I.} Suppose $\sim$ has exactly one equivalence class ${\cls
C}=\{\alpha\omega_1,\dots,\alpha\omega_s\}$, necessarily self-conjugate, with
representative $\alpha$. 
Consider the set of equations $\Eq({\cls C},0)$:
\begin{align*}
(\alpha\omega_1)^n & =\orbitcoeff_1P_1(\alpha\omega_1)+\orbitcoeff_2P_2(\alpha\omega_1)\\
& \vdots\\
(\alpha\omega_s)^n & =\orbitcoeff_1P_1(\alpha\omega_s)+\orbitcoeff_2P_2(\alpha\omega_s)
\end{align*}
For our fixed $r$, the values of $\omega_1^n,\dots,\omega_s^n$ are easy to calculate in
polynomial time, since $\omega_i$ are roots of unity whose order divides $L$.
Then the equations $\Eq({\cls C},0)$ are equivalent to 
\begin{equation}\label{eq: scrunching1}
\left[\begin{array}{c}
\alpha^n\\
\vdots\\
\alpha^n
\end{array}
\right]
=\matr{B}\left[\begin{array}{c}
\orbitcoeff_1\\
\orbitcoeff_2
\end{array}\right],
\end{equation}
where $\matr{B}$ is an $s\times2$ matrix over $\alg$ which, given $r$, is 
computable in polynomial time.
Next we subtract the first row of (\ref{eq: scrunching1}) from rows
$2,\dots,s$, obtaining 
\[
\alpha^n=\orbitphi_1\orbitcoeff_1+\orbitphi_2\orbitcoeff_2
\wedge
\left[\begin{array}{c}
0\\
\vdots\\
0
\end{array}\right]
=\matr{B'}\left[\begin{array}{c}
\orbitcoeff_1\\
\orbitcoeff_2
\end{array}\right].
\]
Here $(\orbitphi_1,\orbitphi_2)$ 
is the first row of the matrix $\matr{B}$, and $\matr{B'}$ is the result 
of subtracting 
$(\orbitphi_1,\orbitphi_2)$
from each of the bottom $s-1$ rows of $\matr{B}$.
Thus, $\Eq({\cls C},0)$ is equivalent to 
$\alpha^n=\orbitphi_1\orbitcoeff_1+\orbitphi_2\orbitcoeff_2$
together with the constraint that 
$(\orbitcoeff_1,\orbitcoeff_2)^T$
must lie in the nullspace of $\matr{B'}$. We now consier the nullspace of $\matr{B'}$.
If its dimension is less than $2$, then we have a linear constraint on
$\orbitcoeff_1,\orbitcoeff_2$. 
This constraint is of the form $\orbitcoeff_1=\chi\orbitcoeff_2$ when the nullspace of $\matr{B'}$
has dimension $1$, and is $\orbitcoeff_1=\orbitcoeff_2=0$ when the nullspace is of dimension $0$.
In both cases, the Master System is equivalent to a lower-dimensional one which
may be computed in polynomial time, so the existence of the bound $N_r$ follows
inductively. 
In the case when the nullspace of $\matr{B'}$ has dimension $2$, then the linear
constraint is vacuous, and $\Eq({\cls C},0)$ is equivalent to
$\alpha^n=\orbitphi_1\orbitcoeff_1+\orbitphi_2\orbitcoeff_2$.

In the same way, for this fixed $r$,
$\Eq({\cls C},1)$ reduces to a single first-derivative equation:
\[ n\alpha^{n-1}=\orbitphi_3\orbitcoeff_1+\orbitphi_4\orbitcoeff_2. \]
We do this for all $\Eq({\cls C},i)$, obtaining a system of equations
equivalent to (\ref{eq: system of equations}) based on the representative of ${\cls C}$,
rather than the actual eigenvalues in ${\cls C}$. Denote the resulting set of
equations by ${\cal F}(\Eq({\cls C}))$.

If some eigenvalue $x\in{\cls C}$ has $\mul(x)\geq3$, then ${\cal
F}(\Eq({\cls C}))$ contains the following triple of equations:
\begin{equation}\label{eq: 2do, 012 tuple}
\left[\begin{array}{c}
\alpha^n\\
n\alpha^{n-1}\\
n(n-1)\alpha^{n-2}
\end{array}\right]
=\orbitcoeff_1\left[\begin{array}{c}
\orbitphi_1\\
\orbitphi_3\\
\orbitphi_5
\end{array}\right]
+\orbitcoeff_2\left[\begin{array}{c}
\orbitphi_2\\
\orbitphi_4\\
\orbitphi_6
\end{array}\right].
\end{equation}
If the vectors on the right-hand side of (\ref{eq: 2do, 012 tuple}) are
linearly independent over $\alg$, then they specify a plane in $\alg^3$, and the
triple states that the point on the left-hand side must lie on this plane. 
Letting
$(A_1,A_2,A_3)^T$ be the normal
of the plane, we obtain 
\begin{align*}
& A_1\alpha^n+A_2n\alpha^{n-1}+A_3n(n-1)\alpha^{n-2}=0 \\
\iff & A_1\alpha^2 + A_2n\alpha + A_3 n(n-1) = 0.
\end{align*}
This is a quadratic equation in $n$. It has at most two roots, both at most
exponentially large in the size of the input, so we just take $N_r$ to be the greater root. 
If the vectors on the right-hand side of (\ref{eq: 2do, 012 tuple}) are linearly dependent 
over $\alg$, then the exponents $n$ which solve (\ref{eq: 2do, 012 tuple}) are precisely
those which solve:
\[
\left[\begin{array}{c}
\alpha^n\\
n\alpha^{n-1}\\
n(n-1)\alpha^{n-2}
\end{array}\right]
=\orbitcoeff_1\left[\begin{array}{c}
\orbitphi_1\\
\orbitphi_3\\
\orbitphi_5
\end{array}\right].
\]
We divide the first equation by the second to obtain 
\[ \frac{\alpha}{n}=\frac{\orbitphi_1}{\orbitphi_3}, \]
which limits $n$ at most one, exponentially large, candidate value $\alpha \orbitphi_3/\orbitphi_1$.

If all eigenvalues $x$ in ${\cls C}$ have $\mul(x)\leq2$ and at least
one has $\mul(x)=2$, then ${\cal F}(\Eq({\cls C}))$ consists of
exactly two equations:
\begin{equation}\label{eq: 2do, 01 tuple}
\left[\begin{array}{c}
\alpha^n\\
n\alpha^{n-1}
\end{array}\right]
=\orbitcoeff_1\left[\begin{array}{c}
\orbitphi_1\\
\orbitphi_3
\end{array}\right]
+\orbitcoeff_2\left[\begin{array}{c}
\orbitphi_2\\
\orbitphi_4
\end{array}\right].
\end{equation}
If $(\orbitphi_1,\orbitphi_3)^T$ and
$(\orbitphi_2,\orbitphi_4)^T$ are linearly
independent over $\alg$, then the right-hand side of (\ref{eq: 2do, 01 tuple}) spans all of
$\alg^2$ as $\orbitcoeff_1,\orbitcoeff_2$ range over $\alg$. Then (\ref{eq: 2do, 01 tuple}) 
is solved by all $n\in\nat$, so we can take $N_r=L$. Otherwise, the
exponents $n$ which solve (\ref{eq: 2do, 01 tuple}) are exactly those which solve
\[
\left[\begin{array}{c}
\alpha^n\\
n\alpha^{n-1}
\end{array}\right]
=\orbitcoeff_1\left[\begin{array}{c}
\orbitphi_1\\
\orbitphi_3
\end{array}\right].
\]
This limits $n$ to at most one candidate value $\alpha \orbitphi_3/\orbitphi_1$, which is exponentially large in
the input size.

Finally, if all eigenvalues $x$ in ${\cls C}$ have $\mul(x)=1$, then
${\cal F}(\Eq({\cls C}))$ contains only the equation
\[ \alpha^n=\orbitcoeff_1\orbitphi_1+\orbitcoeff_2\orbitphi_2, \]
which is solved by all $n\in\nat$ if at least one of $\orbitphi_1,\orbitphi_2$ is non-zero,
and has no solutions if $\orbitphi_1=\orbitphi_2=0$. Either way, we take $N_r=L$ and
are done.

\emph{Case II.} Suppose $\sim$ has exactly two equivalence classes, ${\cls
C}_1$ and ${\cls C}_2$, with respective representatives $\alpha$ and $\beta$, so that 
\[ {\cls C}_1=\{\alpha\omega_1,\dots,\alpha\omega_s\}, \]
\[ {\cls C}_2=\{\beta\omega_1',\dots,\beta\omega_l'\}. \]
The classes could be self-conjugate, in which case $\alpha,\beta\in
\alg\cap\re$, or they could be each other's image under complex
conjugation, in which case $\alpha=\overline{\beta}$. In both cases,
$\alpha/\beta$ is not a root of unity.

As in \emph{Case I,} we transform the system
$\Eq({\cls C}_1)\wedge \Eq({\cls C}_2)$ into the equivalent
system ${\cal F}(\Eq({\cls C}_1))\wedge{\cal F} (\Eq({\cls
C}_2))$.  If all eigenvalues $x$ of $\matr{A}$ have $\mul(x)=1$, then the
resulting system consists of two equations, one for each equivalence class of $\sim$:
\begin{align*}
\alpha^n&=\orbitcoeff_1\orbitphi_1+\orbitcoeff_2\orbitphi_2\\
\beta^n&=\orbitcoeff_1\orbitphi_3+\orbitcoeff_2\orbitphi_4
\end{align*}
If $(\orbitphi_1,\orbitphi_3)^T$ and
$(\orbitphi_2,\orbitphi_4)^T$ are linearly
independent over $\alg$, then there is a solution for each $n$, so just take $N_r=L$.
Otherwise, it suffices to look for $n$ which satisfies
\begin{align*}
\alpha^n &=\orbitcoeff_1\orbitphi_1\\
\beta^n &=\orbitcoeff_1\orbitphi_3
\end{align*}
and hence
\[ \left(\frac{\alpha}{\beta}\right)^n=\frac{\orbitphi_1}{\orbitphi_3}. \]
A bound on $n$ follows from Lemma \ref{lem: algebraic number power problem}.
This argument relies crucially on the fact that $\alpha/\beta$ is not a root of
unity.

If some eigenvalue $x$ of $\matr{A}$ has $\mul(x)\geq2$, say $x\in{\cls C}_1$,
then the system contains the following triple of equations:
\begin{equation}\label{eq: 2do, 010 tuple}
\left[\begin{array}{c}
\alpha^n\\
n\alpha^{n-1}\\
\beta^n
\end{array}\right]
=\orbitcoeff_1\left[\begin{array}{c}
\orbitphi_1\\
\orbitphi_3\\
\orbitphi_5
\end{array}\right]
+\orbitcoeff_2\left[\begin{array}{c}
\orbitphi_2\\
\orbitphi_4\\
\orbitphi_6
\end{array}\right].
\end{equation}
If the vectors on the right-hand side of (\ref{eq: 2do, 010 tuple}) are
linearly dependent over $\alg$, so that the right-hand side describes a space of dimension
$1$, it suffices to look for solutions to 
\[
\left[\begin{array}{c}
\alpha^n\\
n\alpha^{n-1}\\
\beta^n
\end{array}\right]
=\orbitcoeff_1\left[\begin{array}{c}
\orbitphi_1\\
\orbitphi_3\\
\orbitphi_5
\end{array}\right].
\]
Then dividing we obtain 
\[ \frac{\alpha}{n}=\frac{\orbitphi_1}{\orbitphi_3}, \]
which limits $n$ to at most one, exponentially large candidate value $\alpha c_3/c_1$.
Otherwise, if the vectors on the right-hand side of (\ref{eq: 2do, 010 tuple})
are linearly independent over $\alg$, we calculate the normal $(A_1,A_2,A_3)^T$
to the plane described by them and obtain
\[ A_1\alpha^n+A_2n\alpha^{n-1}+A_3\beta^n=0. \]
A bound on $n$ which is exponential in the size of the input follows from Lemma
\ref{lem: skolem 3, one simple and one repeated}. This again relies on the fact
that $\alpha/\beta$ cannot be a root of unity.

\emph{Case III.} Suppose $\sim$ has at least three equivalence classes.  Then
we can choose eigenvalues $\alpha,\beta,\gamma$, each from a distinct
equivalence class, and consider $\eq(\alpha,0)$, $\eq(\beta,0)$
and $\eq(\gamma,0)$:
\[
\left[\begin{array}{c}
\alpha^n\\
\beta^n\\
\gamma^n
\end{array}\right]
=\orbitcoeff_1\left[\begin{array}{c}
P_1(\alpha)\\
P_1(\beta)\\
P_1(\gamma)
\end{array}\right]
+\orbitcoeff_2\left[\begin{array}{c}
P_2(\alpha)\\
P_2(\beta)\\
P_2(\gamma)
\end{array}\right].
\]
If the vectors on the right-hand side are linearly independent over $\alg$, we eliminate
$\orbitcoeff_1,\orbitcoeff_2$ to obtain 
\[ A_1\alpha^n+A_2\beta^n+A_3\gamma^n=0. \]
The left-hand side is a non-degenerate linear recurrence sequence of order $3$,
so a bound on $n$ follows from Lemmas \ref{lem: skolem 3, one dominant root},
\ref{lem: skolem 3, two dominant roots}, \ref{lem: skolem 3, three dominant
roots}.  If the vectors on the right-hand side are not linearly independent over $\alg$,
then we may equivalently consider
\[
\left[\begin{array}{c}
\alpha^n\\
\beta^n\\
\gamma^n
\end{array}\right]
=\orbitcoeff_1\left[\begin{array}{c}
P_1(\alpha)\\
P_1(\beta)\\
P_1(\gamma)
\end{array}\right],
\]
which gives 
\[ \left(\frac{\alpha}{\beta}\right)^n=\frac{P_1(\alpha)}{P_1(\beta)}. \]
An exponential bound on $n$ follows from Lemma \ref{lem: algebraic number power
problem}, because $\alpha/\beta$ is not a root of unity.

Thus, we have now shown that for any $r\in\{0,\dots,L-1\}$, the required bound $N_r$
exists and is at most exponential in the size of the input. Then $N=\max\{N_r : r\in\{0,\dots,L-1\}\}$
exists and is exponentially large, so the  Orbit Problem with two-dimensional target space
is in $\nprp$, by the complexity argument of Section \ref{sec: orbit.outline}.

%% file: orbit.secDim3.tex
\section{Three-dimensional target space}

\label{secThreedim}

Suppose we are given a rational square matrix $\matr{A}$ and polynomials 
$P_1,P_2,P_3$ with rational coefficients such that 
$P_1(\matr{A}),P_2(\matr{A}),P_3(\matr{A})$ are linearly independent over $\rat$. 
We want to decide whether there exists $n\in\nat$ such that
$\matr{A}^n$ lies in the $\rat$-vector space 
$\spa\{P_1(\matr{A}),P_2(\matr{A}),P_3(\matr{A})\}$.
We have derived a Master System of equations (\ref{eq: system of equations})
in the unknowns $(n,\orbitcoeff_1,\orbitcoeff_2,\orbitcoeff_3)$ whose solutions
are precisely the solutions of the matrix equation 
$\matr{A}^n=\orbitcoeff_1 P_1(\matr{A}) + \orbitcoeff_2 P_2(\matr{A}) + \orbitcoeff_3 P_3(\matr{A})$.

In this section, we will show that there exists a bound
$N$, exponentially large in the size of the input, 
such that if the problem instance is positive, then there exists
a witness exponent $n$ with $n<N$. This will be sufficient to show that the 
problem is in the complexity class $\nprp$, as outlined earlier. 

The eigenvalues of $\matr{A}$ may be assumed to be non-zero algebraic
numbers: if $0$ is an eigenvalue, then $\eq(0,0)$ gives a linear
dependence between the coefficients $\orbitcoeff_1,\orbitcoeff_2,\orbitcoeff_3$, 
yielding a lower-dimensional Master System, so the existence of the bound $N$ 
follows inductively.

Following the strategy of the two-dimensional case, we will perform a case analysis
on the residue of $n$ modulo $L$: let $n\bmod L = r$ be fixed throughout this section.
To obtain the required bound $N$, it is sufficient to derive a bound
$N_r$, also exponentially large in the size of the input, such that if there exists
a witness exponent of residue $r$ modulo $L$, then such a witness may be found which 
does not exceed $N_r$. As in the two-dimensional case, we will select tuples of
equations and obtain a bound on $n$ using the results for the  Skolem Problem for
recurrences of order $4$ in Section \ref{appSkolem4}. We will again perform a
case analysis on the equivalence classes of the relation $\sim$. 

\emph{Case I.} Suppose there are at least two pairs of classes 
$({\cls C}_i,\overline{{\cls C}_i})$, $({\cls C}_j,\overline{{\cls C}_j})$ which are
not self-conjugate. Then let $\alpha\in{\cls C}_i$, $\beta=\overline{\alpha}\in
\overline{{\cls C}_i}$, $\gamma\in{\cls C}_j$,
$\delta=\overline{\gamma}\in\overline{{\cls C}_j}$.  Then we consider the tuple
of equations 
\begin{equation}\label{eq: 3do, 1111 tuple}
\left[\begin{array}{c}
\alpha^n\\
\beta^n\\
\gamma^n\\
\delta^n
\end{array}\right]
=\orbitcoeff_1\left[\begin{array}{c}
P_1(\alpha)\\
P_1(\beta)\\
P_1(\gamma)\\
P_1(\delta)
\end{array}\right]
+\orbitcoeff_2\left[\begin{array}{c}
P_2(\alpha)\\
P_2(\beta)\\
P_2(\gamma)\\
P_2(\delta)
\end{array}\right]
+\orbitcoeff_3\left[\begin{array}{c}
P_3(\alpha)\\
P_3(\beta)\\
P_3(\gamma)\\
P_3(\delta)
\end{array}\right].
\end{equation}
If the vectors on the right-hand side are linearly dependent over $\alg$, then we rewrite
the right-hand side as a linear combination of at most two vectors and obtain the required
bound on $n$ by considering a linear recurrence sequence of order $2$ or $3$. 
If the vectors on the right-hand side of (\ref{eq: 3do, 1111
tuple}) are linearly independent over $\alg$, then we calculate the normal of the
three-dimensional subspace of $\alg^4$ that they span, obtaining an
equation 
\begin{equation}\label{eq: 3do, 1111 tuple, skolem form}
A_1\alpha^n+A_2\beta^n+A_3\gamma^n+A_4\delta^n=0
\end{equation}
and hence an exponential bound on $n$ from Lemmas \ref{lem: skolem 4, 1111, not
all same magnitude} and \ref{lem: skolem 4, 1111, same magnitude}. We are
relying on the fact that the ratios of $\alpha,\beta,\gamma,\delta$ are not
roots of unity. Notice that we need $(\alpha,\beta)$ and $(\gamma,\delta)$ to
be pairwise complex conjugates in order to apply Lemma \ref{lem: skolem 4,
1111, same magnitude}. Notice also that we may assume without loss of generality
that $\alpha,\beta,\gamma,\delta$ are algebraic integers, as Lemma 
\ref{lem: skolem 4, 1111, same magnitude} requires. Indeed, as remarked at the
beginning of Section \ref{secReduction}, the input data may be assumed
to be over $\zed$, instead of $\rat$, with the simple technique of scaling the
input by an integer chosen so as to `clear the denominators'. Then $\matr{A}$
is an integer matrix, so its eigenvalues are algebraic integers.

\emph{Case II.} Suppose now that there is exactly one pair of classes 
$({\cls C}_i,\overline{{\cls C}_i})$ which are not self-conjugate. In general, for any
eigenvalue $x$ of $\matr{A}$ we must have
$\mathit{mul}(x)=\mathit{mul}(\overline{x})$.  Therefore, if any eigenvalue
$\alpha\in{\cls C}_i$ has $\mathit{mul}(\alpha)>1$, we can select the tuple of
equations $\eq(\alpha,0)$, $\eq(\alpha,1)$,
$\eq(\overline{\alpha},0)$, $\eq(\overline{\alpha},1)$:
\[
\left[\begin{array}{c}
\alpha^n\\
\overline{\alpha}^n\\
n\alpha^{n-1}\\
n\overline{\alpha}^{n-1}
\end{array}\right]
=\orbitcoeff_1\left[\begin{array}{c}
P_1(\alpha)\\
P_1(\overline{\alpha})\\
P_1'(\alpha)\\
P_1'(\overline{\alpha})
\end{array}\right]
+\orbitcoeff_2\left[\begin{array}{c}
P_2(\alpha)\\
P_2(\overline{\alpha})\\
P_2'(\alpha)\\
P_2'(\overline{\alpha})
\end{array}\right]
+\orbitcoeff_3\left[\begin{array}{c}
P_3(\alpha)\\
P_3(\overline{\alpha})\\
P_3'(\alpha)\\
P_3'(\overline{\alpha})
\end{array}\right].
\]
This gives a non-degenerate linear recurrence sequence of order 4 over
$\alg$ for a recurrence sequence with two repeated characteristic roots:
\[
A_1\alpha^n+A_2\overline{\alpha}^n
+A_3n\alpha^{n-1}+A_4n\overline{\alpha}^{n-1}=0.
\]
An exponential bound $N$ on $n$ follows from Lemma \ref{lem: skolem 4, 22}, since
$\alpha/\overline{\alpha}$ is not a root of unity.

We can now assume that eigenvalues in ${\cls C}_i$ and $\overline{{\cls C}_i}$
contribute exactly one equation to the system. Now we use the fixed value of $r$
to transform
$\Eq({\cls
C}_i)\wedge\Eq (\overline{{\cls C}_i})$ into ${\cal F}
(\Eq({\cls C}_i))\wedge{\cal F} (\Eq(\overline{{\cls C}_i}))$.
Since all eigenvalues in ${\cls C}_i$ and $\overline{{\cls C}_i}$ contribute
one equation each, ${\cal F}(\Eq({\cls C}_i))\wedge{\cal F}(\Eq
(\overline{{\cls C}_i}))$ is just
\begin{align*}
\lambda^n & =\orbitcoeff_1\orbitphi_1+\orbitcoeff_2\orbitphi_2+\orbitcoeff_3\orbitphi_3\\
\overline{\lambda}^n & =\orbitcoeff_1\orbitphi_4+\orbitcoeff_2\orbitphi_5+\orbitcoeff_3\orbitphi_6
\end{align*}
where $\lambda,\overline{\lambda}$ are the representatives of ${\cls C}_i$ and
$\overline{{\cls C}_i}$. We do the same to all self-conjugate classes as well,
reducing the system of equations to an equivalent system based on the representatives of
the equivalence classes, not the actual eigenvalues of $\matr{A}$. This is beneficial,
because the representatives cannot divide to give roots of unity, so we can use 4-tuples
of equations to construct non-degenerate linear recurrence sequences of order
4. 

If there are at least two self-conjugate equivalence classes, with respective
representatives $\alpha,\beta$, we take the tuple 
\begin{align*}
\lambda^n & =\orbitcoeff_1\orbitphi_1+\orbitcoeff_2\orbitphi_2+\orbitcoeff_3\orbitphi_3\\
\overline{\lambda}^{n} & =\orbitcoeff_1\orbitphi_4+\orbitcoeff_2\orbitphi_5+\orbitcoeff_3\orbitphi_6\\
\alpha^n & =\orbitcoeff_1\orbitphi_7+\orbitcoeff_2\orbitphi_8+\orbitcoeff_3\orbitphi_9\\
\beta^n & =\orbitcoeff_1\orbitphi_{10}+\orbitcoeff_2\orbitphi_{11}+\orbitcoeff_3\orbitphi_{12}
\end{align*}
and obtain the following equation, where the left-hand side is a non-degenerate
linear recurrence sequence:
\[ A_1\lambda^n+A_2\overline{\lambda}^n+A_3\alpha^n+A_4\beta^n=0. \]
Then we have an exponentially large bound $N_r$ from Lemmas \ref{lem: skolem 4, 1111, not all same
magnitude} and \ref{lem: skolem 4, 1111, same magnitude}. Similarly, if there
is only one self-conjugate equivalence class, with representative $\alpha$, but some of
its eigenvalues are repeated, we use the tuple
\begin{align*}
\lambda^n &= \orbitcoeff_1\orbitphi_1+\orbitcoeff_2\orbitphi_2+\orbitcoeff_3\orbitphi_3\\
\overline{\lambda}^n &= \orbitcoeff_1\orbitphi_4+\orbitcoeff_2\orbitphi_5+\orbitcoeff_3\orbitphi_6\\
\alpha^n &= \orbitcoeff_1\orbitphi_7+\orbitcoeff_2\orbitphi_8+\orbitcoeff_3\orbitphi_9\\
n\alpha^{n-1} &= \orbitcoeff_1\orbitphi_{10}+\orbitcoeff_2\orbitphi_{11}+\orbitcoeff_3\orbitphi_{12}
\end{align*}
to obtain the non-degenerate instance 
\[ A_1\lambda^n+A_2\overline{\lambda}^n+A_3\alpha^n+A_4n\alpha^{n-1}=0, \]
which gives an exponential bound $N_r$ according to Lemma \ref{lem: skolem 4, 211}.  If
there is exactly one self-conjugate class, with representative $\alpha$, containing no
repeated roots, then the system consists of three equations:
\begin{align*}
\lambda^n&=\orbitcoeff_1\orbitphi_1+\orbitcoeff_2\orbitphi_2+\orbitcoeff_3\orbitphi_3\\
\overline{\lambda}^n&=\orbitcoeff_1\orbitphi_4+\orbitcoeff_2\orbitphi_5+\orbitcoeff_3\orbitphi_6\\
\alpha^n&=\orbitcoeff_1\orbitphi_7+\orbitcoeff_2\orbitphi_8+\orbitcoeff_3\orbitphi_9
\end{align*}
Depending on whether the vectors 
$(\orbitphi_1,\orbitphi_4,\orbitphi_7)^T$,
$(\orbitphi_2,\orbitphi_5,\orbitphi_8)^T$,
$(\orbitphi_3,\orbitphi_6,\orbitphi_9)^T$
are linearly independent over $\alg$, either this triple is solved by all $n\in\nat$ (in which case
set $N_r=L$), or it reduces to a lower-dimensional Master System, yielding the claim inductively. Finally, if
there are no self-conjugate classes, the system consists of only two equations:
\begin{align*}
\lambda^n&=\orbitcoeff_1\orbitphi_1+\orbitcoeff_2\orbitphi_2+\orbitcoeff_3\orbitphi_3\\
\overline{\lambda}^n&=\orbitcoeff_1\orbitphi_4+\orbitcoeff_2\orbitphi_5+\orbitcoeff_3\orbitphi_6
\end{align*}
Again, depending on the dimension of 
\[
\spa\left\{
\left[\begin{array}{c}
\orbitphi_1\\
\orbitphi_4
\end{array}\right],
\left[\begin{array}{c}
\orbitphi_2\\
\orbitphi_5
\end{array}\right],
\left[\begin{array}{c}
\orbitphi_3\\
\orbitphi_6
\end{array}\right]\right\},
\]
we can either set the bound $N_r$ to $L$ (because the transformed Master System is solved by all $n\in\nat$),
or obtain $N_r$ inductively from a lower-dimensional Master System.

\emph{Case III.} All equivalence classes of $\sim$ are self-conjugate. 
The techniques used for this case are identical to the ones already
presented. We use the fixed value of $r$ to reduce to a non-degenerate system
based on the representatives of the classes, with the number of equations
contributed by each class determined by the maximum multiplicity of an
eigenvalue in that class.

If there are less than four equations, then we study the dimension of the
vector space spanned by the vectors on the right-hand side: if it has full
dimension, then we see the Master System is satisfied by all $n$ of the correct
residue $r$, so we can just set $N_r = L$. Otherwise, we obtain the bound
inductively from a lower-dimensional non-degenerate Master System.

On the other hand, if there are at least four equations, then
we can choose four equations which have a solution for $n$ if and only if
an effectively computable non-degenerate LRS of order $4$ vanishes at $n$. 
We then employ the bounds of Theorem~\ref{thm: skolem.bounds} concerning LRS of
order $4$ to obtain the desired $N_r$.

We remark here that it is only for this final case 
that we need the representatives of self-conjugate classes to be real,
necessitating the choice of the magnitude of the eigenvalues in the
class for representative, regardless of whether this magnitude is itself
an eigenvalue. The reason for this technical point is that Lemma
\ref{lem: skolem 4, 1111, same magnitude}, which gives a bound on the index
of zeros of an LRS of order $4$ with four distinct characteristic roots,
requires that the characteristic roots be closed under complex conjugation.
No strengthening of Lemma \ref{lem: skolem 4, 1111, same magnitude} is
known which avoids this precondition -- as we remark in Section \ref{appSkolem4}, 
this is the reason why the  Skolem Problem is open for LRS of order
$4$ over $\alg$. If we had chosen the representative of a self-conjugate
class to be an arbitrary (possibly complex) eigenvalue, we would obtain 
LRS of order $4$ whose characteristic roots do not satisfy the precondition
on Lemma \ref{lem: skolem 4, 1111, same magnitude}, and we would not
be able to obtain our bound $N_r$ here.

%% file: math.secNT.tex
\section{Mathematical techniques}

\subsection{Algebraic numbers: representation and manipulation}\label{appAlgnum}

A complex number $\alpha$ is \emph{algebraic} if there exists a polynomial
$P\in\rats[x]$ such that $P(\alpha)=0$.  The set of algebraic numbers,
denoted by $\alg$, is a subfield of $\mathbb{C}$. The \emph{minimal
polynomial} of $\alpha$ is the unique monic polynomial of least degree which
vanishes at $\alpha$.
The \emph{degree} of
$\alpha\in\alg$ is defined as the degree of its minimal polynomial and is
denoted by $\deg(\alpha)$. The \emph{height} of $\alpha$, denoted by $H_{\alpha}$,
is defined as the
maximum absolute value of the coefficients of the integer polynomial obtained
by scaling the minimal polynomial of $\alpha$ by the least common multiple of
the denominators of its coefficients.
The roots of the minimal polynomial of $\alpha$
(including $\alpha$) are called the \emph{Galois conjugates} of $\alpha$. The
\emph{absolute norm} of $\alpha$, denoted ${\cal N}_{\mathit{abs}}(\alpha)$, is
the product of the Galois conjugates of $\alpha$. By Viete's laws, we have \[
{\cal N}_{\mathit{abs}}(\alpha)=(-1)^{\deg(\alpha)}\frac{a}{b} \] where $a,b$ are
respectively the constant term and the leading coefficient of the minimal polynomial 
of $\alpha$. 
It follows that ${\cal N}_{\mathit{abs}}(\alpha) \in\rat$.  An
\emph{algebraic integer} is an algebraic number whose minimal polynomial
has integer coefficients.
The set of algebraic integers, denoted $\mathcal{O}_\alg$, is a ring under 
the usual addition and multiplication. The algebraic integers are 
\emph{integrally closed}, that is, the roots of any monic polynomial with coefficients
in $\mathcal{O}_\alg$ are all algebraic integers. For any $\alpha\in\alg$, it is
possible to find $\beta\in\mathcal{O}_\alg$ and $m\in\zed$ such that $\alpha=\beta/m$.

The \emph{canonical representation} of an algebraic number $\alpha$ is its
minimal polynomial, along with a numerical approximation of
$\Re(\alpha)$ and $\Im(\alpha)$ of sufficient precision to
distinguish $\alpha$ from its Galois conjugates~\cite[Section 4.2.1]{Coh93}. 
More precisely, we represent $\alpha$ by the tuple 
\[
(P,x,y,R)
\in(\rat[x]\times\rat^3)
\]
meaning that $\alpha$ is the unique root of the irreducible (over $\rats$) polynomial $P$ 
which lies inside the circle
centred at $(x,y)$ in the complex plane with radius $R$. A bound due to
Mignotte \cite{mignotteRootSep} states that for roots $\alpha_j\neq\alpha_k$ of
a polynomial $P(x)$,
\begin{equation}\label{eq:root separation}
|\alpha_j-\alpha_k|>\frac{\sqrt{6}}{d^{(d+1)/2}H^{d-1}},
\end{equation}
where $d$ and $H$ are the degree and height of $P$, respectively.  Thus, if $R$
is restricted to be less than half the root separation bound, the
representation is well-defined and allows for equality checking. Observe that
given its minimal polynomial, the remaining data necessary to describe $\alpha$ is
polynomial in the length of the input. Given $P\in\rat[x]$ and $k\in\nat$, it is 
known how to obtain $k$ bits of the roots of $P$ in time polynomial in the length
of the description of $P$ and in the length of the binary representation of $k$
\cite{panApproximatingRoots}.

When we say an algebraic number $\alpha$ is given, we
assume we have a canonical description of $\alpha$. We will denote by
$\Vert\alpha\Vert $ the length of this description, assuming that integers are
expressed in binary and rationals are expressed as pairs of integers. Observe
that $|\alpha|$ is an exponentially large quantity in $\Vert\alpha\Vert$
whereas $\log|\alpha|$ is polynomially large. Notice also that $1/\log|\alpha|$
is at most exponentially large in $\Vert\alpha\Vert$.  For a rational $a$,
$\Vert a\Vert$ is just the sum of the lengths of its numerator and denominator
written in binary. For a polynomial $P\in\rat[x]$, $\Vert P\Vert$ will
denote $\sum_{j=0}^{\deg(P)}\Vert p_j\Vert$, where
$p_j$ are the coefficients of $P$.
\begin{lemma}\label{lem:operations on algebraic numbers}
Given canonical representations of $\alpha,\beta\in\alg$ and a polynomial
$P\in\rat[x]$, it is possible to compute canonical descriptions of
$\alpha\pm\beta$, $\alpha\beta^{\pm1}$ and $P(\alpha)$ in time polynomial in
the length of the input (that is, in $\Vert\alpha\Vert +\Vert\beta\Vert +\Vert
P\Vert$).
\end{lemma}
\begin{proof}
Let $R,Q$ be the minimal polynomials of $\alpha$ and $\beta$,
respectively. Then the resultant of $R(x-y)$ and $Q(y)$, interpreted as
polynomials in $y$ with coefficients in $\rat[x]$, is a polynomial in $x$
which vanishes at $\alpha+\beta$. We compute it in polynomial time using the
Sub-Resultant algorithm (see Algorithm 3.3.7 in \cite{Coh93}) and factor it
into irreducibles using the LLL algorithm \cite{lll}. Finally, we approximate
the roots of each irreducible factor to identify the minimal polynomial of
$\alpha+\beta$.  The degree of $\alpha+\beta$ is at most $\deg(\alpha)\deg(\beta)$,
while its height is bounded by $H_{\alpha+\beta}\leq H_{\alpha}^{\deg(\alpha)}
H_{\beta}^{\deg(\beta)}$ \cite{zippel}. Therefore, by (\ref{eq:root separation}),
a polynomial number of bits suffices to describe $\alpha+\beta$ unambiguously.
Similarly, we can compute canonical representations of $\alpha-\beta$,
$\alpha\beta$ and $\alpha/\beta$ in polynomial time using resultants, see
\cite{Coh93}.

To calculate $P(\alpha)$ we repeatedly use addition and multiplication. It
suffices to prove that all intermediate results may be represented in
polynomial space. It is clear that their degrees are at most $\deg(\alpha)$, but
it is not obvious how quickly the coefficients of their minimal polynomials
grow. However, there is a simple reason why their representation is
polynomially bounded. Let $\matr{A}$ be the companion matrix of the minimal polynomial of $\alpha$. Then
$P(\alpha)$ is an eigenvalue of $P(\matr{A})$. We can calculate $P(\matr{A})$ using only
polynomial space. Then from the Leibniz formula
\[
\det(\lambda \matr{I}-P(\matr{A}))=
\sum_{\sigma\in S_n}
\mbox{sign}(\sigma)\prod_{i=1}^n(\lambda \matr{I}-P(\matr{A}))_{i,\sigma(i)},
\]
it is evident that the coefficients of the characteristic polynomial of $P(\matr{A})$
are exponentially large in the length of the input, so their representation
requires only polynomial space. This characteristic polynomial may be factored
into irreducibles in polynomial time, so the description of $P(\alpha)$ and of
all intermediate results is polynomially bounded.
\end{proof}
It is trivial to check whether $\alpha=\beta$ and whether $\alpha$ belongs to
one of $\nat,\zed,\rat$. It takes only polynomial time to
determine whether $\alpha$ is a root of unity, and if so, to calculate its
order and phase.

\subsection{Number fields and ideals}\label{appNumfields}

In this section, we recall some terminology and results from algebraic number
theory. For more details, see \cite{Coh93,stewartntall}.  We also define the
ideal-counting function $v_P$, which is a notion of magnitude of algebraic
numbers distinct from the usual absolute value. We follow the presentation of
\cite{TUCS05}.

An \emph{algebraic number field} is a field extension $\mathbb{K}$ of
$\rat$ which, considered as a $\rat$-vector space, has finite
dimension. This dimension is called the \emph{degree} of the number field and
is denoted by $[\mathbb{K}:\rat]$.  
Given two algebraic numbers $\alpha$ and $\beta$, the \emph{Field Membership Problem} 
is to determine whether $\beta \in \rat(\alpha)$ and, if so, to return a polynomial 
$P$ with rational coefficients such that $\beta=P(\alpha)$. 
This problem can be solved using the LLL algorithm, see~\cite[Section 4.5.4]{Coh93}.

For any number field $\mathbb{K}$, there exists an element
$\theta\in\mathbb{K}$ such that $\mathbb{K}=\rat(\theta)$.  Such a
$\theta$ is called a \emph{primitive element} of $\mathbb{K}$ and satisfies
$\deg(\theta)=[\mathbb{K}:\rat]$. The proof
is constructive: there is always a primitive element for 
$\rat(\alpha_1,\alpha_2)$ of the form $\alpha_1 + l\alpha_2$ for some small
integer $l\leq\deg(\alpha_1)\deg(\alpha_2)$. Thus, repeatedly using an algorithm for the Field Membership
Problem for different $l$ is guaranteed to yield a primitive element for 
$\rat(\alpha_1,\alpha_2)$, and therefore by induction, for any number field
$\mathbb{K}=\rat(\alpha_1,\dots,\alpha_k)$ specified by algebraic numbers
$\alpha_1,\dots,\alpha_k$. 
Also using an algorithm for the Field Membership Problem, one can represent
each $\alpha_j$ as a polynomial in $\theta$ and thereby determine a maximal
$\rat$-linearly independent subset of $\{ \alpha_1,\dots,\alpha_k \}$.

There exist exactly $\deg(\theta)$ monomorphisms from
$\mathbb{K}$ into $\mathbb{C}$, given by $\theta\rightarrow\theta_j$, where
$\theta_j$ are the Galois conjugates of the primitive element $\theta$. If $\alpha\in\mathbb{K}$,
then $\deg(\alpha)|\deg(\theta)$. Moreover, if $\sigma_1,\dots,\sigma_{\deg(\theta)}$
are the monomorphisms from $\mathbb{K}$ into $\mathbb{C}$ then
$\sigma_1(\alpha), \dots,\sigma_{\deg(\theta)}(\alpha)$ are exactly the Galois
conjugates of $\alpha$, each repeated $\deg(\theta)/\deg(\alpha)$ times. The
\emph{norm of} $\alpha$ \emph{relative to} $\mathbb{K}$ is defined as 
\[
{\cal N}_{\mathbb{K}/\rat}(\alpha)=
\prod_{j=1}^{\deg(\theta)}\sigma_j(\alpha)= ({\cal
N}_{\mathit{abs}}(\alpha))^{\deg(\theta)/\deg(\alpha)}
\]

For a number field $\mathbb{K}$, the set $\mathcal{O}_\mathbb{K}=
\mathcal{O}_\alg\cap\mathbb{K}$ of algebraic integers in $\mathbb{K}$
forms a ring under the usual addition and multiplication. The ideals of
$\mathcal{O}_\mathbb{K}$ are finitely generated, and form a commutative ring
under the operations 
\[ IJ=[\{ xy \:|\: x\in I,y\in J \}] \]
\[ I+J=\{ x+y\:|\: x\in I,y\in J\}, \]
with unit $\mathcal{O}_\mathbb{K}$ and zero $\{0\}$, where $[S]$ denotes the ideal generated by $S$. 
An ideal $P$ is \emph{prime} if $P=AB$ implies $A=P$ or $A=[1]$.  The fundamental theorem of
ideal theory states that each non-zero ideal may be represented uniquely (up to
reordering) as a product of prime ideals. 

This theorem gives rise to the following \emph{ideal-counting function}
$v_P:\mathcal{O}_\mathbb{K}\backslash\{0\} \rightarrow\nat$.  For a fixed
prime ideal $P$, we define $v_P(\alpha)$ to be the number of times $P$ appears
in the factorisation into prime ideals of $[\alpha]$. That is, 
\[
v_P(\alpha)=k\mbox{ if and only if }
P^k\mid[\alpha]\mbox{ and }P^{k+1}\nmid[\alpha]
\]
We also define $v_P(0)=\infty$. The function satisfies the following
properties:
\begin{itemize}
\item $v_P(\alpha\beta)=v_P(\alpha)+v_P(\beta)$
\item $v_P(\alpha+\beta)\geq\min\{ v_P(\alpha),v_P(\beta)\}$
\item If $v_P(\alpha)\neq v_P(\beta)$, then $v_P(\alpha+\beta)= \min\{v_P(\alpha),v_P(\beta)\}$.
\end{itemize}
For any $\alpha\in\mathbb{K}$ we can find an algebraic integer
$\beta\in\mathcal{O}_\mathbb{K}$ and a rational integer
$n\in\zed\subseteq\mathcal{O}_\mathbb{K}$ such that $\alpha=\beta/n$. We
extend $v_P$ to $\mathbb{K}$ by defining $v_P(\alpha)=v_P(\beta)-v_P(n)$.  The
first of the three properties of $v_P$ above guarantees that this value is
independent of the choice of $\beta,n$, making the extension of $v_P$ to
$\mathbb{K}$ well-defined. Note that the extension preserves the above three
properties.

For an ideal $I\neq\{0\}$, the quotient ring $\mathcal{O}_\mathbb{K}/I$ is
finite. The \emph{norm} of $I$, denoted ${\cal N}(I)$, is defined as
$|\mathcal{O}_\mathbb{K}/I|$. We define also ${\cal N}([0])=\infty$. Notice
that ${\cal N}(I)=1$ if and only if $I=\mathcal{O}_\mathbb{K}$, otherwise
${\cal N}(I)\geq2$.  Each prime ideal $P$ contains a unique prime number $p$,
and ${\cal N}(P)=p^f$ for some natural number $f\geq1$. In general, 
\[
|{\cal N}_{\mathbb{K}/\rat}(\alpha)|=
{\cal N}([\alpha])\geq2^{v_P(\alpha)}
\]
since ${\cal N}(P)\geq2$ for any prime ideal $P$. Hence,
\[
v_P(\alpha)\leq\log_2|{\cal N}_{\mathbb{K}/\rat}(\alpha)|
\leq\log_2|{\cal N}_{\mathit{abs}}(\alpha)|^d
\]
where $d=[\mathbb{K}:\rat]$. Thus, if we are given
$\mathbb{K}=\rat(\alpha_1,\dots,\alpha_k)$ for canonically represented
algebraic numbers $\alpha_j$ and a canonically represented
$\alpha\in\mathbb{K}$, we can observe that $d$ is at most polynomially large in
the length of the input and $|{\cal N}_{\mathit{abs}}(\alpha)|$ is at most
exponentially large in the length of the input. Therefore, $v_P(\alpha)$ is
only polynomially large.

The following lemma is simple, but useful:
\begin{lemma}\label{lem: exists ideal}
Let $\mathbb{K}$ be a number field and $\alpha\in\mathbb{K}$ with
$\alpha\notin\mathcal{O}_\mathbb{K}$.  Then there exists a prime ideal $P$ of
$\mathcal{O}_\mathbb{K}$ such that $v_P(\alpha)\neq0$.
\end{lemma}
\begin{proof}
There exist $\beta\in\mathcal{O}_\mathbb{K}$ and $m\in\zed$ such that
$\alpha=\beta/m$. If $[\beta]=[m]$, then $\beta$ and $m$ are associates, so
$\alpha$ must be a unit of $\mathcal{O}_\mathbb{K}$. Since
$\alpha\notin\mathcal{O}_\mathbb{K}$, it follows that $[\beta]\neq[m]$, so the
factorisations of $[\beta]$ and $[m]$ into prime ideals must differ. Therefore,
$v_P(\beta)\neq v_P(m)$ for some prime ideal $P$, so $v_P(\alpha)\neq0$.
\end{proof}

\subsection{Transcendental number theory}\label{subsec: trnt}

We now move to some techniques from Transcendental Number Theory on
which our results depend in a critical way. 
First, we state a powerful result due to Baker
on linear forms of logarithms of algebraic numbers.
\begin{theorem}\label{thm: bakergeneral}
\cite[Theorem 3.1]{Baker75}
Let $\alpha_1,\ldots,\alpha_m$ be non-zero algebraic numbers
with degrees at most $d$ and heights at most $A$. Further,
let $\beta_0,\dots,\beta_m$ be algebraic numbers with degrees
at most $d$ and heights at most $B\geq2$. Write
\[ \Lambda =\beta_0 + 
\beta_1 \log(\alpha_1) + \ldots + \beta_m \log(\alpha_m) \, . \]
Then either $\Lambda=0$ or $|\Lambda| > B^{-C}$, where $C$ is an
effectively computable number depending only on $m, d, A$ and the
chosen branch of the complex logarithm.
\end{theorem}
Various quantitative versions of this theorem are known with explicit constants,
as well as sharper lower bounds for restricted cases. Of these, in
the present paper we make use of the following result, due to
Baker and W{\"u}stholz, concerning the homogeneous case with
(rational) integer coefficients:
\begin{theorem}\label{thm: bakerwustholz}\cite{BW93}
With the notation as in Theorem \ref{thm: bakergeneral}, 
suppose $\beta_0=0$, $\alpha_1,\dots,\alpha_m\neq 1$,
$\beta_1,\dots,\beta_m\in\zed$ and $A,B\geq e$. Let also
$\log$ be the principal branch of the natural logarithm,
defined by $\log(z) = \log |z|+i\arg(z)$, where $-\pi 
< \arg(z) \leq \pi$. 
Let also $D$ be the degree of the extension field 
$\rats(\alpha_1,\dots,\alpha_m)$ over $\rats$.
Then if $\Lambda\neq0$, then
\[ \log|\Lambda| > -(16mD)^{2(m+2)}(\log(A))^m\log(B). \]
\end{theorem}

The next theorem, due to van der Poorten~\cite{vdp} is analogous to
Baker's bound, but with respect to $P$-adic valuations instead of
the usual Archimedean absolute value.
\begin{theorem}\label{thm: vdPoorten}\cite{vdp}
Let $\alpha_1,\dots,\alpha_m$ be algebraic numbers of degree at most $d$
belonging to a number field $\mathbb{K}$ and with heights at most
$A$. Let $P$ be a prime ideal of $\mathbb{K}$ containing the
rational prime $p$. Let also $\beta_1,\dots,\beta_m$ be rational integers
with absolute values at most $B\geq e^2$. 
If $\alpha_1^{\beta_1}\alpha_2^{\beta_2}\dots\alpha_m^{\beta_m}\neq1$, then
\[
v_P(\alpha_1^{\beta_1}\dots\alpha_m^{\beta_m}-1)
\leq
(16(m+1)d)^{12(m+1)}(p^{d}/\log(p))(\log(A))^m(\log(B))^2.
\]
\end{theorem}

\subsection{Algebraic integers near the unit circle}\label{subsec: blanksby}

Suppose $\alpha\neq 0$ is an algebraic integer. It is easy to see
that it is impossible for all the Galois conjugates of $\alpha$
to be strictly within the unit circle: just notice that the product 
of all Galois conjugates of $\alpha$ must be a non-zero integer by Viete's
laws.
Further, an old result due to Kronecker~\cite{kronecker} establishes that 
unless $\alpha$ is a root of unity, then at least one of its Galois
conjugates must be strictly outside the unit circle.
In this paper, we make use of the following theorem, due to 
Blanksby and Montgomery \cite{BM71}, which strengthens
Kronecker's result by providing an effective separation between this
Galois conjugate and the unit circle.

\begin{theorem}\label{thm: blanksbymontgomery}
Let $\alpha$ be an algebraic integer of degree $d\geq2$.
Then there is a Galois conjugate $\sigma(\alpha)$ of $\alpha$
such that $|\sigma(\alpha)|>1+1/(30d^2\log(6d)).$
\end{theorem}

%% file: math.secLRS.tex
\subsection{Linear recurrence sequences}\label{subsec: LRSprelims}

We now recall some basic properties of linear recurrence sequences.
For more details, we refer the reader to \cite{everest,TUCS05}. 
Let $\field$ be $\re$ or $\cplx$ throughout this section.
A \emph{linear recurrence sequence (LRS)} over $\field$ is an infinite 
sequence $\lrs{u}$ of terms in $\field$ such that there exists a natural number $k$ and
numbers $a_1,\dots,a_k\in\field$ such that $a_k\neq0$ and $\lrs{u}$ satisfies
the linear recurrence equation
\begin{equation}\label{eq: recurrence}
u_{n+k} = a_1 u_{n+k-1} + a_2 u_{n+k-2} + \dots + a_k u_n.
\end{equation}
The recurrence (\ref{eq: recurrence}) is said to have order $k$.  Note
that the same sequence can satisfy different recurrence relations, but
it satisfies a unique recurrence of minimum order.
The \emph{characteristic polynomial} of the sequence $\lrs{u}$ is 
\[ P(x) = x^k - a_1x^{k-1} - a_2x^{k-2} - \dots - a_k \]
and its roots are called the \emph{characteristic roots} of the sequence.

If $\matr{A}\in\field^{k\times k}$ is a square matrix and 
$\vct{v},\vct{w}\in\field^k$ are column vectors, then it can be shown that 
the sequence $u_n=\vct{v^T}\matr{A}^n\vct{w}$ satisfies
a linear recurrence of order $k$. Indeed,
by the Cayley-Hamilton Theorem, $\matr{A}$ satisfies its own 
characteristic equation $\det(\matr{A}-x\matr{I})=0$,
which gives a recurrence relation on $\lrs{u}$ with coefficients
matching those of the characteristic polynomial $\det(\matr{A}-x\matr{I})$ of $\matr{A}$.
Conversely, any LRS may be expressed in this way. Given a linear recurrence 
relation (\ref{eq: recurrence}), it is sufficient to take $\matr{A}$ to be:
\[
\matr{A} = 
\left[
\begin{array}{ccccc}
a_1 & a_2 & \dots & a_{k-1} & a_k \\
1 & 0 & \dots & 0 & 0 \\
0 & 1 & \dots & 0 & 0 \\
\vdots & \vdots & \ddots & \vdots & \vdots \\
0 & 0 & \dots & 1 & 0
\end{array}
\right].
\]
Then if 
$\vct{w}$ is the vector $(u_{k-1},\dots,u_0)^T$ of initial
terms of $\lrs{u}$ in reverse order and $\vct{v}$ is the unit vector
$(0,\dots,0,1)^T$, we have $u_n = \vct{v^T}\matr{A}^n\vct{w}$.
The characteristic polynomial of the LRS is the characteristic polynomial
of $\matr{A}$, and the characteristic roots of the LRS are precisely 
the eigenvalues of $\matr{A}$.

By converting to the Jordan form, from the matrix expression
$u_n = \vct{v^T}\matr{A}^n\vct{w}$ we can obtain a 
closed-form solution for the $n$-th term of the linear recurrence sequence 
in terms of the eigenvalues $\lambda_1,\dots,\lambda_l$ of $\matr{A}$:
\begin{equation}\label{eq: lrscplx}
u_n = \sum_{j=0}^l P_j(n)\lambda_j^n
\end{equation}
for all $n\geq 0$, where $P_j\in\field[x]$ are univariate polynomials of degree strictly
less than the multiplicity of $\lambda_j$ as a root of the characteristic
polynomial of $\matr{A}$.
In the case $\field=\re$, the set of characteristic roots
is closed under complex conjugation.
Thus, if $\rho_1,\dots,\rho_l\in\re$ are the real roots of $P(x)$ and
$\gamma_1,\overline\gamma_1,\dots,\gamma_m,\overline\gamma_m\in\cplx$ are 
the complex ones, the sequence is given by
\begin{equation}\label{eq: lrsreal}
u_n = \sum_{j=1}^l A_j(n)\rho_j^n + 
\sum_{j=1}^m\left( C_j(n)\gamma_j^n + 
\overline{C_j}(n)\overline{\gamma_j}^n \right)
\end{equation}
for all $n\geq0$, where $A_j\in\re[x]$ and $C_j\in\cplx[x]$.
The coefficients of $P_j$ in (\ref{eq: lrscplx}) and
of $A_j,C_j$ in (\ref{eq: lrsreal}) are algebraic numbers, effectively
computable in polynomial time from the description of the LRS.

A linear recurrence sequence is called \emph{degenerate} if for some
pair of distinct characteristic roots $\lambda_1, \lambda_2$ of its
minimum-order recurrence, the ratio $\lambda_1/\lambda_2$ is a root of
unity, otherwise the sequence is \emph{non-degenerate}. As pointed out
in \cite{everest}, the study of arbitrary LRS can be
reduced effectively to that of non-degenerate LRS by partitioning the original LRS
into finitely many non-degenerate subsequences. Specifically, for a
given degenerate linear recurrence sequence $\lrs{u}$ with characteristic
roots $\lambda_j$ and matrix form $u_n=\vct{v^T}\matr{A}^n\vct{w}$, 
let $L$ be the least common multiple of the orders
of all ratios $\lambda_i/\lambda_j$ which are roots of unity. Then
for each $j\in\{0,\dots,L-1\}$, consider the sequence
\[
u^{(j)}_n = \vct{v^T} \matr{A}^{nL+j} \vct{w} 
= \vct{v^T} (\matr{A}^L)^n (\matr{A}^j\vct{w}).
\] 
Each of these sequences has
characteristic roots $\lambda_i^L$ and is therefore non-degenerate.
Indeed, if $\lambda_1^L$ and $\lambda_2^L$ are distinct eigenvalues of
$\matr{A}^L$, and $(\lambda_1/\lambda_2)^L$ is a $k$-th root of unity for
some $k$, then $\lambda_1/\lambda_2$ is an $Lk$-th root
of unity, so $Lk \, | \, \mbox{order}(\lambda_1/\lambda_2)$. But by the definition
of $L$, we also have $\mbox{order}(\lambda_1/\lambda_2)\,|\,L$, so $k=1$. This
yields $\lambda_1^L = \lambda_2^L$, a contradiction.

From the crude lower bound $\varphi(r)\geq\sqrt{r/2}$
on Euler's totient function, it follows that if $\alpha\in\alg$ has degree $d$ and is a primitive
$r$-th root of unity, then $r\leq2d^2$. 
There are $\pin{\matr{A}}$ ratios $\lambda_i/\lambda_j$ to consider, and
if a ratio is a root of unity then its order is $\pin{\matr{A}}$, so 
it follows that $L=\ein{\matr{A}}$.
Thus, non-degeneracy can be ensured by considering at most exponentially
many subsequences of the original LRS.

%% file: skolem.secIntro.tex
\section{ Skolem Problem: introduction}

In the rest of the Appendix, we study the \emph{ Skolem Problem}:
given a linear recurrence sequence (LRS) $\lrs{u}$, determine whether there 
exists a natural number $n$ such that $u_n=0$. The sequence may be real- or
complex-valued, but to make the problem well-defined, we shall require that the
sequence be given in some effective form. For this reason, we take all the linear
recurrence sequences to be over the algebraic numbers, at times 
restricting further to real-valued or rational sequences.

The connection to the Orbit Problem becomes evident
if we recall the matrix representation of linear recurrence sequences. If $\lrs{u}$ is given by
$u_n=\vct{y^T}\matr{A}^n\vct{x}$ for a $d\times d$ matrix $\matr{A}$ and
vectors $\vct{x}$ and $\vct{y}$, then $u_n=0$ if and only if $\matr{A}^n\vct{x}
\in\{\vct{y}\}^\perp$. That is, the orbit of $\vct{x}$ under $\matr{A}$ intersects
the $(d-1)$-dimensional hyperplane $\{\vct{y}\}^\perp$ if and only if the linear
recurrence sequence $\lrs{u}$ of order $d$ contains zero as an element.

The  Skolem Problem has a history dating back to the 1930s, as evidenced
by the celebrated Skolem-Mahler-Lech Theorem, a powerful result which 
characterises the zero sets of linear recurrence sequences: 
\begin{theorem}(Skolem-Mahler-Lech)
Let $\lrs{u}$ be a linear recurrence sequence over a field with characteristic
$0$. Then the zero set of $\lrs{u}$, $Z(u) = \{ n\in\nat : u_n = 0 \}$,
is semilinear, that is, the union
of a finite set and finitely many arithmetic progressions.
\end{theorem} 
This result was originally established in the case of rational LRS in~\cite{skolemSMLthm}, 
then strengthened to include LRS over the algebraic numbers in~\cite{mahlerSMLthm}, 
and finally extended to any field of characteristic 
$0$~\cite{lechSMLthm,MahlerCassels56}. These proofs rely heavily on $p$-adic analysis 
and unfortunately
do not yield constructive methods to compute the zero set of a given linear recurrence,
nor to determine its emptiness. Nonetheless, later work~\cite{BerstelMignotte76}
established an effective procedure to explicitly calculate the arithmetic progressions
mentioned in the theorem for LRS over the rationals. This immediately renders 
decidable the problem of deciding finiteness of the zero set of a rational LRS.
In a similar vein, it is also decidable whether the zero set of a rational LRS 
is equal to $\nat$, and whether it has a finite complement 
\cite[Section II.12]{SalomaaSoittola78}.

Whilst the computation of the infinite component of the zero set
is a significant advancement, no effective method is known to compute 
the finite component or to decide its emptiness. Thus, the decidability 
of the  Skolem Problem remains open and is an outstanding question 
in number theory and theoretical computer science; see, for example,
the exposition of~\cite[Section 3.9]{Tao08}. Efforts towards an upper complexity bound 
have yielded only partial results: decidability for LRS over $\alg$
of order at most $3$ and for LRS over $\ra$ of order $4$ in 
references~\cite{vereshchagin,mignotte}. The decision method
relies crucially on sophisticated results in transcendental number theory,
specifically, Baker's lower bounds on the magnitudes of linear forms in 
logarithms of algebraic numbers and van der Poorten's analogous results 
in the context of $p$-adic valuations. Recently, a proof of decidability 
for LRS of order $5$ was announced in \cite{TUCS05}. However, as pointed 
out in \cite{rp}, the proof incorrectly addresses the case of LRS
of the form:
\[ 
u_n = A\lambda_1^n + \overline{A}\overline{\lambda_1^n} +
B\lambda_2^n + \overline{B}\overline{\lambda_2^n} +
Cr^n,
\]
with one real and four complex characteristic roots with magnitudes satisfying 
$|\lambda_1|=|\lambda_2|>|r|$. Another paper~\cite{Litow97} claims decidability
for all orders, but is also flawed \cite{rp}.

In terms of lower bounds, the strongest known result for the 
Skolem Problem is $\np$-hardness \cite{blondel}. Reference~\cite{Litow97}
claims $\pspace$-hardness, but this has also been shown incorrect~\cite{rp}.

%% file: skolem.secOutline.tex
\section{ Skolem Problem: main result and outline}\label{sec: skolem.outline}

The main technical result of the Appendix of this paper is the following: 
\begin{theorem}\label{thm: skolem.bounds}
Let $\lrs{u}$ be a non-degenerate LRS of order $d$ over $\alg$ which is not identically zero
and whose description has size $\len{u}$.
\begin{enumerate}
\item 
If $d=2$, then there exists a bound $N=\pin{u}$ such that 
if $u_n=0$, then $n<N$.
\item 
If $d=3$, then there exists a bound $N=\ein{u}$ such that 
if $u_n=0$, then $n<N$.
\item 
If $d=4$ and $\lrs{u}$ is over $\ra$, then there exists a 
bound $N=\ein{u}$ such that if $u_n=0$, then $n<N$.
\end{enumerate}
\end{theorem}
References \cite{mignotte,vereshchagin} show the existence
of similar bounds, but make no attempt to quantify them in terms of the
description of the input, thereby showing the problems decidable, but obtaining
no more specific complexity upper bound.
Our contribution is to show the bounds are at most exponential
in the size of the input, and in fact, polynomial for LRS of order $2$. This permits
us to obtain the following complexity bounds for the  Skolem Problem for rational LRS:
\begin{theorem}\label{thm: skolem.complexity}
For LRS over $\rat$ of order at most $4$, the  Skolem Problem is in the complexity class
$\nprp$. Further, for LRS over $\rat$ of order $2$, the problem is in $\ptime$.
\end{theorem}

Two points need to be addressed: how to reduce from arbitrary LRS to non-degenerate, non-zero
LRS, and how to obtain the complexity results of Theorem \ref{thm: skolem.complexity}
from the bounds of Theorem \ref{thm: skolem.bounds}.

On the first point, as we showed in Section \ref{subsec: LRSprelims}, the study of arbitrary LRS
can be reduced effectively to the non-degenerate case. This uses the technique of partitioning
a given LRS into $L$ non-degenerate subsequences, where 
\begin{equation}\label{eq: skolem.lcmsize}
L = \mbox{lcm}\{ \mbox{order}(\lambda_i/\lambda_j)\,:\,
\lambda_i,\lambda_j\mbox{ characteristic roots and }
\lambda_i/\lambda_j\mbox{ root of unity } 
\}.
\end{equation}
Specifically, if $\lrs{u}$ is given, then we consider the sequences $\lrs{v^{(j)}}$
defined by $v_n^{(j)}=u_{Ln+j}$ for $j=0,\dots,L-1$. These subsequences are non-degenerate, 
so for the purposes of
showing decidability, non-degeneracy may be assumed without loss of generality.
However, when attempting to establish a more precise complexity upper bound,
the size of $L$ needs to be taken into account.

Recall that if $\alpha$ is an algebraic number of degree $d$ and a 
root of unity of order $r$, then $r\leq 2d^2$. In particular, if
$\lrs{u}$ is an LRS defined by $u_n=\vct{x^T}\matr{A}^n\vct{y}$ described
using $\len{u}=\len{\vct{x}}+\len{\matr{A}}+\len{\vct{y}}$ bits, and 
$\lambda_i,\lambda_j$ are characteristic roots such that $\lambda_i/\lambda_j$
is a root of unity, then $\mbox{order}(\lambda_i/\lambda_j)=\pin{u}$.
Since (\ref{eq: skolem.lcmsize}) takes the least common multiple of
the orders of $\bigoh(\len{u}^2)$ ratios 
$\lambda_i/\lambda_j$ and each order is polynomially large in 
the size of the input, it follows that $L=\ein{u}$. Moreover, this is
not an over-approximation: it is easy to construct LRS where the ratios
$\lambda_i/\lambda_j$ are roots of unity of mutually coprime orders, thereby making
$L$ at least exponential in the size of the input. Thus, for arbitrary
LRS, applying this technique to eliminate non-degeneracy carries an
exponential overhead.

However, in this paper, we restrict our attention to LRS of order at most $4$. 
Therefore, in (\ref{eq: skolem.lcmsize}),
the number of ratios considered is bounded by an absolute constant,
so $L$ is the least common multiple of a fixed number of polynomially large
orders, hence $L=\pin{u}$. 

Furthermore, if the LRS is over $\rats$, then the degree of each characteristic root 
$\lambda_i$ is at most $4$, since we know \emph{a priori} that the characteristic 
polynomial of the sequence has rational coefficients.
Then the degrees of all ratios $\lambda_i/\lambda_j$ are also absolutely bounded, so
$L=\bigoh(1)$. Therefore, in our context of rational LRS
of bounded order, non-degeneracy may be obtained by considering a constant
number of subsequences, whose matrix representation may be computed from that
of $\lrs{u}$ in polynomial time.

Some of the non-degenerate subsequences $\lrs{v^{(j)}}$ could potentially be identically zero, resulting
in the zero set of the original degenerate sequence $\lrs{u}$ containing an entire arithmetic
progression. For each subsequence, we check directly whether this is 
the case by examining its first $d$ terms, and if so, then the problem instance is 
immediately positive. Otherwise, the instance has been reduced to a polynomially-large
(or constant, in the rational case) number of problem instances,
each featuring a non-degenerate LRS which is not the zero sequence.

The second point is how to obtain the complexity upper bounds.
For non-degenerate rational LRS $\lrs{u}$ defined by $u_n=\vct{x^T}\matr{A}^n\vct{y}$, 
let $N$ denote the bound provided by Theorem \ref{thm: skolem.bounds}. If
the sequence is of order $2$, then $N$ is only polynomial in the size of the
input. Thus, we simply calculate $u_n$ for all $n<N$. All intermediate results
are rational numbers, and we only ever raise $\matr{A}$ to a polynomially
large power, so the representation of all intermediate results stays polynomially
bounded. Thus, the $\ptime$ upper bound for LRS of order $2$ follows directly.

For rational LRS of order $3$ or $4$, the bound $N$ is at most exponential in the size 
of the input. We argue the problem is in $\npeqslp$,
where $\eqslp$ is the complete class for the following problem:
given a division-free straight-line program (or equivalently, an arithmetic
circuit) producing an integer $M$, determine whether $M=0$.
Since the bound $N$ is at most exponentially large in the size of the input, 
an $\np$ algorithm can guess the index of a purported zero: $n\in\nat$ with $n<N$. 
Thus, we only need to verify
that $u_n=0$. Direct calculation is not an option, since $n$ is
exponential in the size of the input, whilst the entries of $\matr{A}^n$ are
doubly-exponential in magnitude, requiring an exponential number of
bits to write down. However, we can easily represent the entries of
$\matr{A}^n$ as polynomially-sized arithmetic circuits, using the 
technique of repeated squaring. Then verifying $u_n=0$ reduces to checking
whether a polynomially-large arithmetic circuit evaluates to $0$, which
can be solved by an $\eqslp$ oracle. The bound $\npeqslp$
follows directly. Finally, it is known that
$\eqslp\subseteq\corp$ \cite{schonhage}, so we
also have membership in $\nprp$, as Theorem \ref{thm: skolem.complexity} claims.

Therefore, all that remains is to prove Theorem \ref{thm: skolem.bounds}. We devote
the rest of the Appendix to the technical details of the proof. 
Sections \ref{appSkolem2}, \ref{appSkolem3} and \ref{appSkolem4} address LRS of order
$2$, $3$ and $4$, respectively. Section \ref{appBaker} shows two applications of
Baker's Theorem which are crucially important for orders $3$ and $4$.

%% file: skolem.secOrder2.tex
\section{LRS of order two}\label{appSkolem2}

In this section, we consider the problem of whether a 
linear recurrence sequence $\lrs{u}$ of order $2$ over $\alg$ contains zero as a term.
The characteristic equation of the recurrence may have one repeated 
root $\theta$, or two distinct roots $\theta_1,\theta_2$. Thus, the
$n$-th term of the sequence is given by one of the following:
\begin{align}
u_n & =(A+Bn)\theta^n & &
\mbox{ (where $A,B,\theta\in\alg$ and $B,\theta\neq 0$)}
\label{eq: skolem2rep} \\
u_n & =A\theta_1^n+B\theta_2^n & &
\mbox{ (where $A,B,\theta_1,\theta_2\in\alg$ and $A,B,\theta_1,\theta_2\neq 0$)}
\label{eq: skolem2simple}
\end{align}
Solving the Skolem Problem for LRS of the form (\ref{eq: skolem2rep}) is trivial: 
simply determine whether the unique root of $A+Bx$ is a natural number.
We therefore concentrate on LRS of the form (\ref{eq: skolem2simple}).
In this case, $u_n=0$ if and only if $(\theta_1/\theta_2)^n=-B/A$. 

Thus, the problem
reduces to the \emph{algebraic number power problem:} decide whether there exists
$n\in\nat$ such that 
\begin{equation}\label{eq: algebraic number power}
\alpha^n=\beta
\end{equation}
for given $\alpha,\beta\in\alg$. The assumption of non-degeneracy of $\lrs{u}$
allows us to assume $\alpha$ is not a root of unity\footnote{Notice in passing 
that if $\alpha$ is a root of unity,
then the algebraic number power problem is easy to decide: simply determine whether $\beta$ 
is an $r$-th root of unity, where $r$ is the order of $\alpha$.
If this is indeed the case, however, then there exists no bound of the kind promised by Theorem
\ref{thm: skolem.bounds}, since $\alpha^n=\beta$ holds periodically.}. 
The algebraic number power problem is decidable \cite{TUCS05}. 
Reference \cite{orbit} proved a polynomial bound
on $n$ when $\beta$ has the form $P(\alpha)$ for a given $P\in\rat[x]$. 
We give a brief recapitulation of the decidability proof of \cite{TUCS05} 
and sharpen it to extract a polynomial bound on $n$.


\begin{lemma}\label{lem: algebraic number power problem}
Suppose $\alpha,\beta\in\alg$. If $\alpha$ is not a root of unity, then
there exists a bound $N$ such that if (\ref{eq: algebraic number
power}) holds, then $n<N$. Moreover, $N=\pin{I}$, where 
$\len{I} = \len{\alpha} +\len{\beta}$ is the length of the input.
\end{lemma}

\begin{proof}
Let $\mathbb{K}=\mathbb{Q}(\alpha,\beta)$. If $\alpha$ is not an algebraic
integer, then by Lemma \ref{lem: exists ideal} there exists a prime ideal $P$
in the ring $\mathcal{O}_\mathbb{K}$ such that $v_P(\alpha)\neq0$. Then if
$\alpha^n=\beta$, we have 
\[ v_P(\alpha^n)=nv_P(\alpha)=v_P(\beta). \]
If $v_P(\alpha)$ and $v_P(\beta)$ have different signs, then we are done.
Otherwise, 
\[
n=\frac{v_P(\beta)}{v_P(\alpha)}\leq|v_P(\beta)|
\leq\log_2|{\cal N}_{\mathbb{K}/\mathbb{Q}}(\beta)|
\leq\log_2|{\cal N}_{\mathit{abs}}(\beta)|^d,
\]
where $d=[\mathbb{Q}(\alpha,\beta):\mathbb{Q}]$ is at most polynomially large
in $\Vert\alpha\Vert +\Vert\beta\Vert$.  It follows that the bound on $n$ is
polynomially large in the length of the input.

Suppose therefore that $\alpha$ is an algebraic integer. It is not a root of unity
by the premise of the Lemma, so by Theorem \ref{thm: blanksbymontgomery}
(Blanksby and Montgomery),
$\alpha$ has a Galois conjugate
$\sigma(\alpha)$ such that 
\[ |\sigma(\alpha)|>1+\frac{1}{30d^2\log(6d)}, \]
where $d$ is the degree of $\alpha$.
This implies 
\[ \frac{1}{\log|\sigma(\alpha)|}<60d^2\log(6d). \]
Recall that $\sigma$ is a monomorphism on $\mathbb{K}$, so 
$\sigma(\alpha^n) = (\sigma(\alpha))^n$.
Then if $\alpha^n=\beta$, we have 
\[
n=\frac{\log|\sigma(\beta)|}{\log|\sigma(\alpha)|}
<\log|\sigma(\beta)|60d^2\log(6d).
\]
Observe that if we are given canonical descriptions of $\alpha$ and $\beta$,
then $60d^2\log(6d)$ is at most polynomially large in
$\Vert\alpha\Vert$, and $\log|\sigma(\beta)|$ is at most polynomially large in
$\Vert\beta\Vert$. It follows that the bound on $n$ is polynomial in the length
of the input. 
\end{proof}


%% file: skolem.secBaker.tex
\section{Application of Baker's Theorem}\label{appBaker}

Before we proceed to LRS of order $3$ and $4$, we make a brief
diversion to show two pertinent applications of Baker's Theorem. They
essentially capture the technically difficult core of the Skolem
Problem for LRS of order $3$ or $4$, so for clarity, they are
exhibited here first, prior to their use in the context of LRS.

The first application concerns powers $\lambda^n$ ($n\in\nat$) of an algebraic 
number $\lambda$ on the unit circle. We show that for large $n$ and any fixed
$b\in\alg$, the distance $|\lambda^n-b|$ cannot be `too small', unless $\lambda$
is a root of unity.

\begin{lemma}\label{lem: baker application}
Let $\lambda,b\in\alg$, where $|\lambda|=1$ and $\lambda$ is not a root
of unity. Suppose $\phi(n)$ is a function from $\nat$ to $\cplx$ for
which there exist $a,\chi\in\rat$ such that $\chi\in(0,1)$ and $|\phi(n)|\leq
a\chi^n$. There exists a bound $N$ such that if
\begin{equation}\label{eq: baker situation}
\lambda^n=\phi(n)+b,
\end{equation}
then $n<N$. Moreover, $N$ is at most exponential in the length of the input
$\len{I}=\len{\lambda}+\len{b}+\len{a}+\len{\chi}$.
\end{lemma}

\begin{proof}
The left-hand side of (\ref{eq: baker situation}) describes points on the unit
circle, whereas the right-hand side tends to $b$. If $|b|\neq1$, then for $n$
large enough, the right-hand side of (\ref{eq: baker situation}) will always be
off the unit circle.  This happens when 
\[ n>\frac{\log(||b|-1|/a)}{\log(\chi)}. \]

The difficult case is when $b$ is on the unit circle. Here we will use Baker's
Theorem to derive a bound on $n$. Consider the angle $\Lambda$ between
$\lambda^n$ and $b$. Since $\lambda$ is not a root of unity, 
by Lemma~\ref{lem: algebraic number power problem}, this 
angle can be zero for at most one value of $n$, which is polynomially large in $\len{I}$.
Otherwise, write
\[ \Lambda=\log\frac{\lambda^n}{b}=n\log(\lambda)-\log(b)+ 2k_n\log(-1)\neq0, \]
where $k_n$ is an integer chosen so that $\Lambda=i\tau$ for some
$\tau\in[0,2\pi)$. Then $2n$ is an upper bound on the height of the
coefficients in front of the logarithms (because $k_n\leq n$), $H=\max\{
H_{\lambda},H_b,3\}$ is a height bound for the arguments to the logarithms and
$d=\max\{\deg(\lambda),\deg(b)\}$ is a bound on the degrees. 
Then by Theorem \ref{thm: bakerwustholz} (Baker-W\"ustholz), we have
\[ \log|\Lambda|>-(48d)^{10}\log^2H\log(2n), \]
which is equivalent to
\[ |\Lambda|>(2n)^{-(48d)^{10}\log^2H}. \]
This is a lower bound on the length of the arc between $\lambda^n$
and $b$. The length of the chord is at least half of the bound: 
$|\lambda^n-b|\geq|\Lambda|/2$.
So in the equation $\lambda^n-b=\phi(n)$, the left-hand side is bounded below
by an inverse polynomial in $n$. However, the right-hand side shrinks
exponentially quickly in $n$. For all $n$ large enough, the right-hand side will
be smaller in magnitude than the left-hand side. 

We will now quantify the bound on $n$. Let $p_1=(48d)^{10}\log^2H$ and
$p_2=2$. Observe that $p_1,p_2=\pin{I}$.
Then (\ref{eq: baker situation}) cannot hold if 
\[ \frac{1}{2}(p_2n)^{-p_1}\geq a\chi^n, \]
which is equivalent to 
\[ -\log(2)-\log(a)-p_1\log(p_2)-p_1\log(n)\geq n\log(\chi). \]
Define $p_3=\log(2)+\log(a)+p_1\log(p_2)$ and $p_4=\max\{p_3,p_1\}=\pin{I}$. 
Then it suffices to have
\[ \frac{p_4}{-\log(\chi)}\leq\frac{n}{1+\log(n)}, \]
which is guaranteed by 
\[ \sqrt{n}\geq\frac{p_4}{-\log(\chi)}. \]
Observe that $-1/\log(\chi)$ is at most exponentially large in $\Vert\chi\Vert$.
Therefore, the bound on $n$ is exponential in the size of the input.
\end{proof}


Continuing in the same line, we next consider two algebraic numbers, 
$\lambda_1$ and $\lambda_2$, whose powers define discrete trajectories 
embedded in two circles in the complex plane:
$a\lambda_1^n$ and $b\lambda_2^n + c$ as $n$ varies over $\nat$. 
The following lemma shows that unless $\lambda_1,\lambda_2$ are roots of
unity, then for large $n$, the $n$-th points of the two trajectories are 
never `too close' to each other.

\begin{lemma}\label{lem: two intersecting circles}
Suppose $\lambda_1,\lambda_2,a,b,c\in\alg$ are non-zero, where
$|\lambda_1|=|\lambda_2|=1$ and $\lambda_1,\lambda_2$ are not roots of unity.
Let $\phi(n)$ be a function from $\nat$ to $\cplx$ such that
$0<|\phi(n)|\leq w\chi^n$ for some $w,\chi\in\rat$, $\chi\in(0,1)$. Then
there exists a bound $N$ such that if 
\begin{equation}\label{eq: two intersecting circles}
a\lambda_1^n=b\lambda_2^n+c+\phi(n),
\end{equation}
then $n<N$. Moreover, $N=\ein{I}$, where $\len{I}=\len{\lambda_1}+
\len{\lambda_2}+\len{a}+\len{b}+\len{c} + \len{w} + \len{\chi}$.
\end{lemma}

\begin{proof}
Multiplying the equation by $\overline{c}/|c||a|$ allows us to assume that
$|a|=1$ and $c\in\re^+$. 

Let $f(n)=a\lambda_1^n$, $g(n)=b\lambda_2^n+c$.  It is clear that $f(n)$
describes points on the unit circle~${\cal O}_1$, whilst $g(n)$ describes
points on the circle ${\cal O}_2$ with centre $c$ on the real line and 
radius~$|b|$. 

If these circles do not intersect, then for $n$ large enough, $|\phi(n)|$ will
be forever smaller than the smallest distance between the circles.  This
happens when 
\[ n>\frac{\log(c-|b|-1)-\log(w)}{\log(\chi)}, \]
which is an exponential lower bound on $n$ in the size of the input.

Suppose now the circles intersect in two points, $z_1$ and $z_2$.  Let $\mathcal{L}_1$ be
the horizontal line through $z_1$ and $\mathcal{L}_2$ the horizontal line through $z_2$.
Let $\mathcal{L}_1\cap\mathcal{O}_1=\{ x_1,z_1\}$, $\mathcal{L}_1\cap\mathcal{O}_2=\{y_1,z_1\}$,
$\mathcal{L}_2\cap\mathcal{O}_1=\{x_2,z_2\}$ and $\mathcal{L}_2\cap\mathcal{O}_2=\{y_2,z_2\} $. It is
trivial that $z_2=\overline{z_1}$, $x_2=\overline{x_1}$, $y_2=\overline{y_1}$.

\input{tikzFigure.tex}

We first argue that for $n$ large enough, (\ref{eq: two intersecting circles})
can hold only if for some intersection point $z_i$, $\Re(z_i)$ lies between
$\Re(f(n))$ and $\Re(g(n))$, or $\Im(z_i)$ lies between $\Im(f(n))$ and $\Im(g(n))$.
This can only be violated in two symmetric situations: either 
\begin{enumerate}
\item $f(n)$ is on the
arc $z_1z_2$ of ${\cal O}_1$ which lies inside ${\cal O}_2$ and $g(n)$ is on
the arc $y_1y_2$ of ${\cal O}_2$ which lies outside ${\cal O}_1$, or 
\item $f(n)$ is
on the arc $x_1x_2$ of ${\cal O}_1$ which lies outside ${\cal O}_2$ and $g(n)$
is on the arc $z_1z_2$ of ${\cal O}_2$ which lies inside ${\cal O}_1$. 
\end{enumerate}
In the first situation, when $g(n)$ is on the arc $y_1y_2$ of ${\cal O}_2$ outside
${\cal O}_1$, we have
\[ |f(n)-g(n)|\geq|g(n)|-1\geq|y_1|-1. \]
Since the point $y_1$ is strictly to the right of $1$ on the complex plane,
this lower bound is positive, and moreover it is independent of $n$, so
(\ref{eq: two intersecting circles}) cannot hold for $n$ large enough because $\phi(n)$ tends to zero
exponentially quickly. In particular, (\ref{eq: two intersecting circles}) does not
hold if
\[ n>\frac{\log(|y_1|-1)-\log(w)}{\log(\chi)}, \]
which is exponentially large in the size of the input. The second situation is
analogous.

Therefore, we can assume that one of the intersection points $z_i$ separates
$f(n)$ and $g(n)$ horizontally or vertically in the figure. That is, $z_i$
satisfies $\Re(f(n))\leq \Re(z_i) \leq \Re(g(n))$ or $\Im(f(n))\leq \Im(z_i) \leq
\Im(g(n))$.  We will show a lower bound on $|f(n)-g(n)|$ which shrinks slower
than exponentially. The real (horizontal) and imaginary (vertical) cases are
completely analogous. We show the working for the real case. Assume that
$\Re(z_i)$ lies between $\Re(f(n))$ and $\Re(g(n))$. Clearly, 
\[ |f(n)-g(n)|\geq|\Re(g(n)-f(n))|=|\Re(z_i-f(n))|+|\Re(g(n)-z_i)|. \]
Let $\alpha=\arg(\lambda_1)$, $\gamma=\arg(a)$ and $\beta=\arg(z_i)$. Then 
\[
|\Re(z_i-f(n))|=|\cos(n\alpha+\gamma)-\cos(\beta)|
=2\left|\sin\frac{\beta-n\alpha-\gamma}{2}
\sin\frac{\beta+n\alpha+\gamma}{2}\right|.
\]
Let $u_n,v_n$ be appropriately chosen integers so that
\[ 
\frac{\beta-n\alpha-\gamma}{2}+u_n\pi\in
\left[-\frac{\pi}{2},\frac{\pi}{2}\right],
\]
\[ 
\frac{\beta+n\alpha+\gamma}{2}+v_n\pi\in
\left[-\frac{\pi}{2},\frac{\pi}{2}\right].
\]
Then using the inequality
\[ 
|\sin(x)|\geq\frac{|x|}{\pi}
\mbox{ for }x\in\left[-\frac{\pi}{2},\frac{\pi}{2}\right],
\]
we have
\[
\left|\sin\frac{\beta-n\alpha-\gamma}{2}\right|
\geq\frac{1}{\pi}\left|\frac{\beta-n\alpha-\gamma}{2}+\pi u_n\right|,
\]
\[
\left|\sin\frac{\beta+n\alpha+\gamma}{2}\right|\geq
\frac{1}{\pi}\left|\frac{\beta+n\alpha+\gamma}{2}+\pi v_{n}\right|.
\]
Both of these expressions are sums of logarithms of algebraic numbers, 
non-zero for $n$ exceeding a polynomially large bound in $\len{I}$ by
Lemma~\ref{lem: algebraic number power problem},
so we can give lower bounds for them using Theorem~\ref{thm: bakerwustholz}
(Baker-W\"ustholz) as in Lemma~\ref{lem: baker application}: 
\[ |\Re(z_i-f(n))|\geq(p_1n)^{-p_2} \]
for some $p_1,p_2>0$ which are independent of $n$ and at most polynomially
large in the input. A similar lower bound holds for 
$|\Re(g(n)-z_i)|$.
If $\delta=\arg(\lambda_2)$, $\eta=\arg(b)$ and $\theta=\arg(z_i-c)$, we have
\[ |\Re(g(n)-z_i)|=|b|(\cos(n\delta+\eta)-\cos(\theta))\geq(p_3n)^{-p_4}, \]
where $p_3,p_4>0$ are independent of $n$ and have at most polynomial size in
the input. Hence we have 
\[ |f(n)-g(n)|\geq2(p_5n)^{-p_6}, \]
where $p_5=\max\{p_1,p_3\}$ and $p_6=\max\{p_2,p_4\}$.  Since $\phi(n)$ shrinks
exponentially quickly, a bound on $n$ follows past which (\ref{eq: two
intersecting circles}) cannot hold. In the manner of Lemma \ref{lem: baker
application}, we can show that this bound is exponentially large in the size of
the input.  The vertical case is analogous, except that considering imaginary
parts gives sines instead of cosines, so we shift all angles by $\pi/2$ and
proceed as above. If the circles are tangent and neither lies inside the other,
then the intersection point separates $f(n)$ and $g(n)$ horizontally, so we are
done by the above analysis. 

Finally, suppose that the circles are tangent and one lies inside the other:
$|b|+c=1$. The argument of $f(n)$ is $\gamma+n\alpha$. By the law of cosines
applied to the triangle with vertices $f(n)$ and the centres of the circles, we
have 
\[ |f(n)-c|^2=c^2+1-2c\cos(\gamma+n\alpha). \]
Therefore, the shortest distance from $f(n)$ to a point on ${\cal O}_2$ is 
\[ h(n)=\sqrt{c^2+1-2c\cos(\gamma+n\alpha)}-(1-c). \]
Let $A(n)=\sqrt{c^2+1-2c\cos(\gamma+n\alpha)}$ and $B=1-c$. Since $A\leq1+c$,
we have $A+B\leq2$, so 
\[ h(n)=A-B=\frac{A^2-B^2}{A+B}\geq c(1-\cos(\gamma+n\alpha)) \]
Let $k_n$ be an integer, so that 
\[ \gamma+n\alpha+k_n2\pi\in[-\pi,\pi). \]
By Lemma~\ref{lem: algebraic number power problem}, this is zero for at most 
one, polynomially large in $\len{I}$, value of $n$. For larger $n$,
a lower bound on this angle follows from Theorem \ref{thm: bakerwustholz} (Baker-W\"ustholz):
\[ |\gamma+n\alpha+k_n2\pi|\geq(p_7n)^{-p_8} \]
for some constants $p_7,p_8>0$ which are polynomially large in the input. Then 
\[ \cos(\gamma+n\alpha)\leq\cos((p_7n)^{-p_8}), \]
so
\[ h(n)\geq c(1-\cos((p_7n)^{-p_8})). \]
From the Taylor expansion of $\cos(x)$, it follows easily that 
\[ 1-\cos(x)\geq\frac{11}{24}x^2\mbox{ for }x\leq1. \]
Since $p_7,p_8\geq1$, we have $(p_7n)^{-p_8}\leq1$.  Therefore, 
\[ h(n)\geq c\frac{11}{24}(p_7n)^{-2p_8}. \]
This lower bound on $h(n)$ shrinks inverse-polynomially as $n$ grows. Recall
that $h(n)$ is the smallest distance from $f(n)$ to ${\cal O}_2$. It follows
that for $n$ large enough, $|\phi(n)|<h(n)$ forever, so $f(n)=g(n)+\phi(n)$
cannot hold.  In the manner of Lemma \ref{lem: baker application}, we can show
that the bound on $n$ is exponentially large in the input.
\end{proof}


%% file: tikzFigure.tex
\begin{tikzpicture} [scale=1.6]
\coordinate (c) at (2.8, 0);
\draw [->, thick, name path=yaxis] (0, -1.5) -- (0, 1.5);
\draw [->, thick, name path=xaxis] (-1.3, 0) -- (6, 0);
\draw [thick, name path=O1] (0, 0) circle [radius=1];
\draw [thick, name path=O2] (c) circle [radius=2];

\path [name intersections={of=O1 and O2}] ;

\coordinate (z1) at (intersection-1) {};
\coordinate (z2) at (intersection-2) {};
\node [label={[label distance=0.1]45:{\small $z_1$}}] at (z1) {};
\node [label={[label distance=0.1]315:{\small $z_2$}}] at (z2) {};

\draw let \p1 = (z1) in 
[thick, green, name path=L1] (-1.3,\y1) -- (6, \y1);
\draw let \p2 = (z2) in
[thick, green, name path=L2] (-1.3,\y2) -- (6, \y2);

\path [name intersections={of=L1 and O1}] ;
\coordinate (x1) at (intersection-2); 
\node [label={[label distance=0.1]135:{\small $x_1$}}] at (x1) {}; 

\path [name intersections={of=L1 and O2}] ;
\coordinate (y1) at (intersection-1);
\node [label={[label distance=0.1]45:{\small $y_1$}}] at (y1) {};

\path [name intersections={of=L2 and O1}] ;
\coordinate (x2) at (intersection-1) ;
\node [label={[label distance=0.1]225:{\small $x_2$}}] at (x2) {};

\path [name intersections={of=L2 and O2}] ;
\coordinate (y2) at (intersection-2) ;
\node [label={[label distance=0.1]315:{\small $y_2$}}] at (y2) {};

\fill (z1) circle [radius=0.04] ;
\fill (z2) circle [radius=0.04] ;
\fill (x1) circle [radius=0.04] ;
\fill (x2) circle [radius=0.04] ;
\fill (y1) circle [radius=0.04] ;
\fill (y2) circle [radius=0.04] ;
\fill (c) circle [radius=0.04] ;

\node [above] at (2.8, 0) {\small $c$} ; 

\path [name path=proj1] (c) -- ($(x1)!(c)!(y1)$) ;
\path [name intersections={of=proj1 and L1}] ;
\node [above] at (intersection-1) {$\mathcal{L}_1$} ;

\path [name path=proj1] (c) -- ($(x2)!(c)!(y2)$) ;
\path [name intersections={of=proj1 and L2}] ;
\node [below] at (intersection-1) {$\mathcal{L}_2$} ;

\coordinate (O1label) at (1.3,1.7) ;
\node at (O1label) {$\mathcal{O}_2$} ;
\coordinate (O2label) at (-0.7, 1) ;
\node at (O2label) {$\mathcal{O}_1$} ;
\end{tikzpicture}

%% file: skolem.secOrder3.tex
\section{LRS of order three}\label{appSkolem3}

We now move to the problem of determining whether a 
linear recurrence sequence $\lrs{u}$ of order $3$ over $\alg$ contains zero as an element.
The characteristic equation of such a sequence may have
either three distinct (real or complex) roots, or one 
repeated real root and one simple real root, or one real root of multiplicity $3$.
Thus, the $n$-th element of the sequence is given by one of the following:
\begin{align}
u_n & =A\alpha^n+B\beta^n+C\gamma^n & &
\mbox{ (where $A,B,C,\alpha,\beta,\gamma\in\alg$ are all non-zero)}
\label{eq: skolem.rec111} \\
u_n & =(A+Bn)\alpha^n+C\beta^n & &
\mbox{ (where $A,B,C,\alpha,\beta\in\alg$ with $B,C,\alpha,\beta\neq0$)}
\label{eq: skolem.rec21} \\
u_n & = (Cn^2 + Bn + A)\alpha^n & &
\mbox{ (where $A,B,C,\alpha\in\alg$ with $C,\alpha\neq0$)}
\label{eq: skolem.rec3}
\end{align}
Finding the zeros of LRS of the form (\ref{eq: skolem.rec3}) is trivial:
simply check whether the quadratic $Cn^2 + Bn + A$ has roots which are
natural numbers. Thus, we focus on the remaining two possibilities.
We will consider only non-degenerate sequences: the ratios of the roots
$\alpha,\beta,\gamma$ are not roots of unity.

First we consider $\lrs{u}$ given by (\ref{eq: skolem.rec111}).
Notice that $A,B,C,\alpha,\beta,\gamma$ are all non-zero, otherwise 
the sequence satisfies a recurrence relation of lower order. 
Thus, we can rearrange $u_n=0$ to obtain:
\begin{equation}\label{eq: skolem 3 rearranged}
\left(\frac{\beta}{\alpha}\right)^{n}
=-\frac{C}{B}\left(\frac{\gamma}{\alpha}\right)^{n}-\frac{A}{B}.
\end{equation}
Assume without loss of generality $|\alpha|\geq|\beta|\geq|\gamma|$.  
In Lemmas \ref{lem: skolem 3, one dominant root}, \ref{lem:
skolem 3, two dominant roots}, \ref{lem: skolem 3, three dominant roots} below,
we consider separately the cases $|\alpha|>|\beta|$, $|\alpha|=|\beta|>|\gamma|$
and $|\alpha|=|\beta|=|\gamma|$, and obtain a bound on $n$ which is exponential
in the length of the description of the sequence and beyond which $u_n=0$ cannot
hold.


\begin{lemma}\label{lem: skolem 3, one dominant root}
Suppose $\lrs{u}$ is given by (\ref{eq: skolem.rec111}).
If $|\alpha|>|\beta|$, then there exists a bound $N$ such that if
$u_n=0$, then $n<N$. Moreover, $N=\ein{I}$, where $\len{I}$ is the
length of the input $\len{A}+\len{B}+\len{C}+\len{\alpha}+\len{\beta}+\len{\gamma}$.
\end{lemma}
\begin{proof}
This follows straightforwardly from the dominance of $\alpha$. If 
\[
n>\max\left\{ \frac{\log|A/2B|}{\log|\beta/\alpha|},
\frac{\log|A/2C|}{\log|\gamma/\alpha|}\right\},
\]
then 
\[
\left|
-\frac{B}{A}\left(\frac{\beta}{\alpha}\right)^n
-\frac{C}{A}\left(\frac{\gamma}{\alpha}\right)^n
\right|
\leq
\left|\frac{B}{A}\left(\frac{\beta}{\alpha}\right)^n\right|
+
\left|\frac{C}{A}\left(\frac{\gamma}{\alpha}\right)^n\right|<
\frac{1}{2}+\frac{1}{2}=1.
\]
\end{proof}


\begin{lemma}\label{lem: skolem 3, two dominant roots}
Suppose $\lrs{u}$ is given by (\ref{eq: skolem.rec111}).
If $|\alpha|=|\beta|>|\gamma|$, then there exists a bound $N$ such that if
$u_n=0$, then $n<N$. Moreover, $N=\ein{I}$, where $\len{I}$ is the
length of the input $\len{A}+\len{B}+\len{C}+\len{\alpha}+\len{\beta}+\len{\gamma}$.
\end{lemma}
\begin{proof}
This is a direct application of Lemma \ref{lem: baker application} to equation
(\ref{eq: skolem 3 rearranged}). 
\end{proof}


\begin{lemma}\label{lem: skolem 3, three dominant roots} 
Suppose $\lrs{u}$ is given by (\ref{eq: skolem.rec111}).
If $|\alpha|=|\beta|=|\gamma|$, there exist at most two values of $n$ such
that $u_n=0$. Moreover, they are at most exponential in the length of the input 
$\len{A}+\len{B}+\len{C}+\len{\alpha}+\len{\beta}+\len{\gamma}$
and are computable in polynomial time.
\end{lemma}
\begin{proof}
The left-hand side of (\ref{eq: skolem 3 rearranged}) as a function of $n$
describes points on the unit circle in the complex plane, whereas the
right-hand side describes points on a circle centred at $-A/B$ with radius
$|C/B|$. Note these circles do not coincide, because $A\neq0$. We can obtain
their equations and compute their intersection point(s). If they do not
intersect, then equation (\ref{eq: skolem 3 rearranged}) can never hold. Otherwise, the
equation can only hold if the two sides are simultaneously equal to the same
intersection point. For each of the (at most two) intersection points $\theta$,
let 
\[
S_1=\left\{ n\:\left|\:
\left(\frac{\beta}{\alpha}\right)^n=\theta\right.\right\},
\]
\[
S_2=\left\{ n\:\left|\:
-\frac{C}{B}\left(\frac{\gamma}{\alpha}\right)^n
-\frac{A}{B}=\theta\right.\right\}.
\]
Observe that $|S_i|\leq1$, because $\beta/\alpha$ and $\gamma/\alpha$ are not
roots of unity. We compute $S_1$ and $S_2$ from the bound in Lemma \ref{lem:
algebraic number power problem} and check whether $S_1\cap S_2$ is non-empty.
\end{proof}


Next, we consider LRS of the form (\ref{eq: skolem.rec21}).
We will assume that $B,C,\alpha,\beta$ are all non-zero, otherwise the sequence
satisfies a linear recurrence of lower order.

\begin{lemma}\label{lem: skolem 3, one simple and one repeated}
Suppose $\lrs{u}$ is given by (\ref{eq: skolem.rec21}).
There exists a bound $N$ such that if $u_n=0$,
then $n<N$. Moreover, $N=\ein{I}$, where $\len{I}$ is
the length of the input $\len{A}+\len{B}+\len{C}+\len{\alpha}+\len{\beta}$.
\end{lemma}
\begin{proof}
We wish to solve for $n\in\nat$ the equation:
\begin{equation}\label{eq: skolem 3, repeated root}
(A+Bn)\alpha^n+C\beta^n=0.
\end{equation}
If $|\alpha|\geq|\beta|$, then for 
\[ n>\frac{|A|+|C|}{|B|}, \]
we have 
\[ |C|<|B|n-|A|\leq|A+Bn|, \]
so
\[ |C\beta^n|<|(A+Bn)\alpha^n|, \]
therefore (\ref{eq: skolem 3, repeated root}) cannot hold.  Now suppose
$|\alpha|>|\beta|$ and rewrite (\ref{eq: skolem 3, repeated root}) as 
\[ \frac{A+Bn}{C}=-\left(\frac{\beta}{\alpha}\right)^n. \]
Equation (\ref{eq: skolem 3, repeated root}) implies
\[
\left|\frac{\beta}{\alpha}\right|^n=\left|\frac{A+Bn}{C}\right|
\leq \left|\frac{A}{C}\right| + \left|\frac{B}{C}\right|n.
\]
However, we will show that for all $n$ large enough, this fails to hold.
Indeed, the inequality
\[
\left|\frac{\beta}{\alpha}\right|^{n}
>
\left|\frac{A}{C}\right|+\left|\frac{B}{C}\right|n
\]
is implied by 
\[ d\left(n+1\right)<\left|\frac{\beta}{\alpha}\right|^n, \]
where $d=\max\{|A/C|,|B/C|\}$.
Taking logarithms, we see that it suffices to have 
\[ \frac{n}{1+\log(n+1)}>\frac{f}{\log|\beta/\alpha|}, \]
where $f=\max\{\log(d),1\}$. Noting that $1+\log(n+1)<2\sqrt{n}$ for all
$n\geq1$, we see that it suffices to have 
\[ n>4f^2/\log^2|\beta/\alpha| \]
to guarantee that (\ref{eq: skolem 3, repeated root}) cannot hold.  This is an
exponential bound on $n$ in the length of the input. 
\end{proof}

%% file: skolem.secOrder4.tex
\section{LRS of order four}\label{appSkolem4}

We now proceed to the problem of determining whether a 
linear recurrence sequence $\lrs{u}$ of order $4$ over $\alg$ contains zero as an element.
As before, we assume non-degeneracy of the sequence.
Depending on the roots of the characteristic
polynomial, the $n$-th term of the sequence is given by one of the following
(where $A,B,C,D,\alpha,\beta,\gamma,\delta$ are algebraic):
\begin{align}
u_n & =A\alpha^n+B\beta^n+C\gamma^n+D\delta^n & &
\mbox{ (where $A,B,C,D\neq0$)}
\label{eq: skolem.rec1111} \\
u_n & =(A+Bn)\alpha^n+C\beta^n+D\gamma^n & &
\mbox{ (where $B,C,D\neq0$)}
\label{eq: skolem.rec211} \\
u_n & =(A+Bn)\alpha^n+(C+Dn)\beta^n & &
\mbox{ (where $B,D\neq0$)}
\label{eq: skolem.rec22} \\
u_n & =(A+Bn+Cn^2)\alpha^n+D\beta^n & &
\mbox{ (where $C,D\neq0$)}
\label{eq: skolem.rec31} \\
u_n & =(A+Bn+Cn^2+Dn^3)\alpha^n & &
\mbox{ (where $D\neq0$)}
\label{eq: skolem.rec4}
\end{align}

Solving $u_n=0$ in the case of $\lrs{u}$ given by (\ref{eq: skolem.rec4})
is trivial: just calculate canonical descriptions of the roots of 
$A+Bx+Cx^2+Dx^3$ and check whether any are natural numbers.

In the case of $\lrs{u}$ given by (\ref{eq: skolem.rec31}),
rearrange $u_n=0$ as 
\[ (A+Bn+Cn^2)\left(\frac{\alpha}{\beta}\right)^n=-D. \]
The left-hand side tends to $0$ or $\infty$ in magnitude, depending on whether
$|\alpha|<|\beta|$. In both cases, since $C,D\neq0$, a bound on $n$
follows which is at most exponential in the size of the input. 

The remaining three cases, where $\lrs{u}$ is of the form (\ref{eq: skolem.rec1111}),
(\ref{eq: skolem.rec211}) or (\ref{eq: skolem.rec22}) are more involved. They are
the subject of Lemmas \ref{lem: skolem 4, 22}, \ref{lem: skolem 4, 211}, 
\ref{lem: skolem 4, 1111, not all same magnitude} and 
\ref{lem: skolem 4, 1111, same magnitude}, which show the existence of 
a bound $N$ which is at most exponentially large in the size of the input
and beyond which $u_n=0$ cannot hold.

Note that for complex algebraic LRS given by (\ref{eq: skolem.rec1111}) with 
characteristic roots all of the same magnitude, the  Skolem Problem is 
not known to be decidable.
Thus, our final technical result, Lemma \ref{lem: skolem 4, 1111, same magnitude} 
will require the simplifying assumption that $u_n\in\ra$ for all~$n$. This 
is the only reason why Theorem \ref{thm: skolem.bounds} insists that LRS of order $4$
be real algebraic. In all other cases, as shown by Lemmas \ref{lem: skolem 4, 22}, 
\ref{lem: skolem 4, 211} and \ref{lem: skolem 4, 1111, not all same magnitude},
an exponential bound on $n$ exists even for complex algebraic LRS.


\begin{lemma}\label{lem: skolem 4, 22}
Suppose $\lrs{u}$ is non-degenerate and is given by (\ref{eq: skolem.rec22}).
There exists a bound $N=\ein{I}$ such that if $u_n=0$, then $n<N$,
where $\len{I}$ is the length of the input 
$\len{A}+\len{B}+\len{C}+\len{D}+\len{\alpha}+\len{\beta}$.
\end{lemma}
\begin{proof}
We wish to solve for $n\in\nat$ the equation:
\begin{equation}\label{eq: skolem 4, 22}
(A+Bn)\alpha^n+(C+Dn)\beta^n=0\mbox{ (where $B,D\neq0$)}.
\end{equation}
Rearrange (\ref{eq: skolem 4, 22}) as 
\begin{equation}\label{eq: skolem 4, 22 rearranged}
\lambda^n=-\frac{(C+Dn)}{(A+Bn)},
\end{equation}
where $\lambda=\alpha/\beta$ is not a root of unity. 
The right-hand side of (\ref{eq: skolem 4, 22
rearranged}) tends to $-D/B$ as $n$ tends to infinity.

If $\lambda$ is an algebraic integer, then by Theorem 
\ref{thm: blanksbymontgomery} (Blanksby and Montgomery), 
it has a Galois conjugate $\sigma(\lambda)$ such that 
\[ |\sigma(\lambda)|>1+\frac{1}{30d^2\log(6d)}, \]
where $d$ is the degree of $\lambda$.
Assume the monomorphism $\sigma$ has been applied to both sides of 
(\ref{eq: skolem 4, 22 rearranged}), so $|\lambda|$ is 
bounded away from $1$ by an inverse polynomial in the size of the input.
By the triangle inequality, if 
\[ n \geq \frac{|BC|+|AD|+|AB|}{|B|^2} \defn N_1=\ein{I}, \]
then
\[ \left|\frac{C+Dn}{A+Bn}\right| \leq \frac{|D|n+|C|}{|B|n-|A|} \leq 
\left|\frac{D}{B}\right|+1. \]
Following the reasoning of Lemma \ref{lem: algebraic number power problem} and
relying on the Blansky and Montgomery bound, we see there exists a bound
$N_2\in\pin{I}$ such that if $n>N_2$, then $|\lambda^n|>|D/B|+1$. Therefore,
for $n>\max\{N_1,N_2\}=\ein{I}$, equation (\ref{eq: skolem 4, 22 rearranged}) 
cannot hold.

Second, suppose $\lambda$ is not an algebraic integer. Then by Lemma \ref{lem:
exists ideal} there exists a prime ideal $P$ in the ring of integers of
$\mathbb{K}=\mathbb{Q}(\alpha,\beta,A,B,C,D)$ such that $v_P(\lambda)\neq0$.
Without loss of generality, we can assume $v_P(\lambda)>0$ (if
$v_P(\lambda)<0$, swap $\alpha$ with $\beta$, $A$ with $C$, and $B$ with $D$).
Applying $v_P$ to (\ref{eq: skolem 4, 22 rearranged}) gives 
\begin{eqnarray*}
v_P(\lambda^n) & = & nv_P(\lambda) \\
& = & v_P\left(-\frac{C+Dn}{A+Bn}\right)  \\
& \leq & \log\left|{\cal N}_{\mathbb{K}/\mathbb{Q}}
\left(-\frac{C+Dn}{A+Bn}\right)\right| \\ 
& \leq &  
[\mathbb{K}:\mathbb{Q}]\log\left|{\cal N}_\mathit{abs}
\left(-\frac{C+Dn}{A+Bn}\right)\right| \\
& = & [\mathbb{K}:\mathbb{Q}]\log
\prod_{i=1}^{[\mathbb{K}:\mathbb{Q}]} 
\left|\frac{\sigma_i(C)+\sigma_i(D)n}{\sigma_i(A)+\sigma_i(B)n}\right|, \\ 
\end{eqnarray*}
where $\sigma_1,\dots,\sigma_{[\mathbb{K}:\mathbb{Q}]}$ are the 
monomorphisms from $\mathbb{K}$ into $\mathbb{C}$. As in the previous case,
if 
\[ 
n>\frac{|\sigma_i(BC)|+|\sigma_i(AD)|
+|\sigma_i(AB)|}{|\sigma_i(B)|^2}
\defn N_i=\ein{I},
\]
then we have
\[
\left|\frac{\sigma_i(C)+\sigma_i(D)n}{\sigma_i(A)+\sigma_i(B)n}\right|
\leq
\left|\frac{\sigma_i(D)}{\sigma_i(B)}\right| + 1
\defn e_i
=
\ein{I}.
\]
It follows therefore that if $n>\max_i\{N_i\}$, we have
\[ v_P\left(-\frac{C+Dn}{A+Bn}\right)
\leq
[\mathbb{K}:\mathbb{Q}]\sum_{i=1}^{[\mathbb{K}:\mathbb{Q}]}\log e_i 
\defn M = 
\pin{I}. \]
Then for $n>\max_i\{N_i\}$ and $n>M$, we have 
\[ v_P(\lambda^n) = nv_P(\lambda) \geq n > M, \]
whereas 
\[ v_P\left(-\frac{C+Dn}{A+Bn}\right) \leq M, \]
so equation (\ref{eq: skolem 4, 22 rearranged}) cannot hold. 
\end{proof}


\begin{lemma}\label{lem: skolem 4, 211}
Suppose $\lrs{u}$ is non-degenerate and is given by (\ref{eq: skolem.rec211}).
There exists a bound $N=\ein{I}$ such that if $u_n=0$, then $n<N$,
where $\len{I}$ is the length of the input 
$\len{A}+\len{B}+\len{C}+\len{D}+\len{\alpha}+\len{\beta}+\len{\gamma}$.
\end{lemma}
\begin{proof}
We wish to solve for $n\in\nat$ the equation:
\begin{equation}\label{eq: skolem 4, 211}
(A+Bn)\alpha^n+C\beta^n+D\gamma^n=0\mbox{ (where $B,C,D\neq0$)}.
\end{equation}
First suppose $|\alpha|\geq|\beta|,|\gamma|$.  Then the term $(A+Bn)\alpha^n$
is dominant. More precisely, rewrite (\ref{eq: skolem 4, 211}) as 
\[
A+Bn=-C\left(\frac{\beta}{\alpha}\right)^n
-D\left(\frac{\gamma}{\alpha}\right)^n
\]
and observe that if 
\[ n>\frac{|A|+|C|+|D|}{|B|}, \]
then 
\[
|A+Bn|\geq|B|n-|A|>|C|+|D|\geq\left|-C\left(\frac{\beta}{\alpha}\right)^n
-D\left(\frac{\gamma}{\alpha}\right)^n\right|,
\]
so (\ref{eq: skolem 4, 211}) cannot hold due to the strictness of the above
inequality.

Second, suppose that $|\beta|>|\alpha|,|\gamma|$.  Then the term $C\beta^n$ is
dominant. More precisely, rewrite (\ref{eq: skolem 4, 211}) as
\begin{equation}\label{eq: edits.1}
(A+Bn)\left(\frac{\alpha}{\beta}\right)^n
+D\left(\frac{\gamma}{\beta}\right)^n=-C.
\end{equation}
We show that for $n$ sufficiently large, the inequalities
\[ \left|D\left(\frac{\gamma}{\beta}\right)^n\right|<\frac{|C|}{2} \]
and
\[
\left|\left(A+Bn\right)\left(\frac{\alpha}{\beta}\right)^n\right|
<\frac{|C|}{2}
\]
both hold, rendering (\ref{eq: edits.1}) impossible.  The former inequality
holds for $n>\log|C/2D|/\log|\gamma/\beta|$, which is at most exponentially large
in the input. The latter inequality is implied by
\[
\left|(n+1)\left(\frac{\alpha}{\beta}\right)^n\right|
<\frac{|C|}{2M},
\]
where $M=\max\{|A|,|B|\}$. Now let $r=\lceil -\log(2)/\log(\alpha/\beta)\rceil$,
so that 
\[ \left(\frac{\alpha}{\beta}\right)^r\leq\frac{1}{2}, \]
and consider only $n$ of the form $n=kr$ for $k\in\mathbb{Z}^+$.  If 
\[ k>\frac{\log|C/4Mr|}{\log(7/8)} \]
and $k\geq5$, we have 
\[
\left(\frac{\alpha}{\beta}\right)^{kr}k<
\left(\frac{1}{2}\right)^{k}(k+1)<
\left(\frac{7}{8}\right)^{k}<
\frac{|C|}{4Mr},
\]
so 
\[
\left(\frac{\alpha}{\beta}\right)^n(n+1)
\leq\left(\frac{\alpha}{\beta}\right)^n2n
<\frac{|C|}{2M}.
\]
It is clear that $r$ is at most exponentially large in the size of the input,
whereas the bound on $k$ is polynomial. Therefore, the bound on $n$ is
exponential.

Finally, suppose $|\beta|=|\gamma|>|\alpha|$.  Rewrite (\ref{eq: skolem 4,
211}) as
\[
\left(\frac{\beta}{\gamma}\right)^n
=-\frac{D}{C}-\frac{A+Bn}{C}\left(\frac{\alpha}{\gamma}\right)^n.
\]
Then an exponential bound on $n$ follows from Lemma \ref{lem: baker
application}, because the right-hand side is a constant plus an exponentially
decaying term, whereas the left-hand side is on unit circle. 
\end{proof}


\begin{lemma}\label{lem: skolem 4, 1111, not all same magnitude}
Suppose $\lrs{u}$ is non-degenerate and is given by (\ref{eq: skolem.rec1111}). 
Suppose that
$\alpha,\beta,\gamma,\delta$ do not all have the same magnitude.
There exists a bound $N=\ein{I}$ such that if $u_n=0$, then $n<N$,
where $\len{I}$ is the length of the input 
$\len{A}+\len{B}+\len{C}+\len{D}+\len{\alpha}+\len{\beta}+\len{\gamma}+\len{\delta}$.
\end{lemma}
\begin{proof}
We wish to solve for $n\in\nat$ the equation:
\begin{equation}
\label{eq: skolem 4, 1111}
A\alpha^n+B\beta^n+C\gamma^n+D\delta^n=0\mbox{ (where $A,B,C,D\neq0$)}.
\end{equation}
Let $|\alpha|\geq|\beta|\geq|\gamma|\geq|\delta|$.  First, if
$|\alpha|>|\beta|$, then $A\alpha^n$ is the dominant term in (\ref{eq: skolem
4, 1111}). Rewrite the equation as 
\[
\frac{B}{A}\left(\frac{\beta}{\alpha}\right)^n
+\frac{C}{A}\left(\frac{\gamma}{\alpha}\right)^n
+\frac{D}{A}\left(\frac{\delta}{\alpha}\right)^n=-1
\]
and observe that if 
\[
n>\max\left\{ 
\frac{\log|3B/A|}{\log|\alpha/\beta|},
\frac{\log|3C/A|}{\log|\alpha/\gamma|},
\frac{\log|3D/A|}{\log|\alpha/\delta|}
\right\},
\]
then 
\[
\left|
\frac{B}{A}\left(\frac{\beta}{\alpha}\right)^n
+\frac{C}{A}\left(\frac{\gamma}{\alpha}\right)^n
+\frac{D}{A}\left(\frac{\delta}{\alpha}\right)^n
\right|
<\frac{1}{3}+\frac{1}{3}+\frac{1}{3}=1.
\]
Second, if $|\alpha|=|\beta|>|\gamma|$, then rewrite (\ref{eq: skolem 4, 1111})
as
\begin{equation}\label{eq: skolem 4, 1111 rearranged}
\left(\frac{\beta}{\alpha}\right)^n
=-\frac{A}{B}
-\frac{C}{B}\left(\frac{\gamma}{\alpha}\right)^n
-\frac{D}{B}\left(\frac{\delta}{\alpha}\right)^n.
\end{equation}
The left-hand side of (\ref{eq: skolem 4, 1111 rearranged}) is on the unit
circle, whereas the right is a constant plus exponentially decaying terms. An
exponential bound on $n$ follows from Lemma \ref{lem: baker application}. 

Finally, if $|\alpha|=|\beta|=|\gamma|>|\delta|$, then an exponential bound on
$n$ follows from Lemma \ref{lem: two intersecting circles} applied to equation
(\ref{eq: skolem 4, 1111 rearranged}). 
\end{proof}


Thus, the only outstanding problem is to solve $u_n=0$ in the case 
of $\lrs{u}$ given by (\ref{eq: skolem.rec1111}) when
$|\alpha|=|\beta|=|\gamma|=|\delta|$. This case is difficult for general algebraic
$\alpha,\beta,\gamma,\delta$: it is in fact the reason why the  Skolem Problem
is open for LRS of order $4$ over $\alg$. However, for real LRS, the set of
characteristic roots is closed under complex conjugation, so complex
roots come in conjugate pairs. 

Another simplifying observation which is helpful for this last outstanding case is
that for any LRS $\lrs{u}$ over $\alg$, one can find another LRS
$\lrs{v}$ over $\mathcal{O}_\alg$ such that $u_n=0$ if and only if $v_n=0$.
Indeed, recall that for any algebraic number $\alpha$, it is possible to find an algebraic integer
$\beta$ and a rational integer $M$ such that $\alpha = \beta/M$: it is sufficient to choose $M$
to be the least common multiple of all denominators of the coefficients of the minimal
polynomial of $\alpha$. 
Then suppose the sequence $\lrs{u}$ has initial
terms $u_0,\dots,u_{d-1}\in\alg$ and satisfies a recurrence equation
$u_{n} = \sum_{j=0}^{d-1} a_j u_{n-j-1}$ with $a_0,\dots,a_{d-1}\in\alg$.
Let $M\in\zed$ be chosen so that $Ma_j\in\mathcal{O}_\alg$ and $Mu_j\in\mathcal{O}_\alg$
for $j=0,\dots,d-1$. Then it is easy to see that the sequence $\lrs{v}$ defined by 
$v_n = M^{n+1}u_n$ has the same zero set as $\lrs{u}$, has algebraic integer initial terms 
and satisfies a linear recurrence relation of order $d$ with algebraic integer coefficients.
Since $M$ can be written down using only polynomial space, this reduction to the 
integer case can be carried out in polynomial time. Therefore, by the integral closure 
of $\mathcal{O}_\alg$, we can assume the characteristic roots $\alpha,\beta,\gamma,\delta$ 
are algebraic integers.

With these two observations in place, we proceed to the final technical result
concerning the  Skolem Problem for LRS of order $4$ over $\ra$:

\begin{lemma}\label{lem: skolem 4, 1111, same magnitude}
Suppose $\lrs{u}$ is non-degenerate and is given by (\ref{eq: skolem.rec1111}). 
Suppose that $\alpha,\beta,\gamma,\delta$ are algebraic integers with
$|\alpha|=|\beta|=|\gamma|=|\delta|$. Suppose also 
$\{\alpha,\beta,\gamma,\delta\}$ is closed under complex conjugation.
There exists a bound $N=\ein{I}$ such that if $u_n=0$,
then $n<N$, where $\len{I}$ is the length of the input
$\len{A}+\len{B}+\len{C}+\len{D}+\len{\alpha}+\len{\beta}+\len{\gamma}+\len{\delta}$.
\end{lemma}
\begin{proof}
Let $\mathbb{K}=\mathbb{Q}(\alpha,\beta,\gamma,\delta,A,B,C,D)$. We have to solve
for $n\in\nat$ the equation:
\begin{equation}
\label{eq: skolem 4, 1111second}
A\alpha^n+B\beta^n+C\gamma^n+D\delta^n=0\mbox{ (where $A,B,C,D\neq0$)}.
\end{equation}
The closure of $\{\alpha,\beta,\gamma,\delta\}$ under complex conjugation,
the equality $|\alpha|=|\beta|=|\gamma|=|\delta|$ and the non-degeneracy 
of the LRS imply that the characteristic roots are two pairs of complex 
conjugates, so assume without loss of generality that $\beta=\overline{\alpha}$ 
and $\gamma=\overline{\delta}$. If
$\alpha/\beta$ is an algebraic integer, then since it is not a root of unity,
there exists a monomorphism $\sigma$ from $\mathbb{K}$ to $\mathbb{C}$ such
that $|\sigma(\alpha)|\neq|\sigma(\beta)|$.  Applying $\sigma$ to (\ref{eq:
skolem 4, 1111second}) leads to a Skolem instance of order 4 with roots whose
magnitudes are not all the same.  A bound on $n$ follows from Lemma \ref{lem:
skolem 4, 1111, not all same magnitude}.

Suppose then that $\alpha/\beta$ is not an algebraic integer. By the reasoning
of Lemma \ref{lem: exists ideal}, there exists a prime ideal $P$ in
$\mathcal{O}_\mathbb{K}$ such that $v_P(\alpha)\neq v_P(\beta)$ and at least
one of $v_P(\alpha)$ and $v_P(\beta)$ is strictly positive. Assume without loss
of generality that 
\[ v_P(\alpha)>v_P(\beta)\geq0. \]
Since $\alpha\beta=\gamma\delta=|\alpha|^2$, we have
\[ v_P(\alpha)+v_P(\beta)=v_P(\gamma)+v_P(\delta). \]
Therefore, at most two of the roots are smallest under the valuation $v_P$. 

If one root, say $\beta$, is strictly smaller under $v_P$ than the rest, then
rewrite (\ref{eq: skolem 4, 1111second}) as 
\begin{equation}\label{eq: final, case 1}
A\alpha^n+B\beta^n=-C\gamma^n-D\delta^n
\end{equation}
Since $v_P(\beta)<v_P(\alpha)$, for $n>v_P(A/B)/v_P(\beta/\alpha)$ we have 
\[ v_P(A\alpha^n+B\beta^n)=v_P(B)+nv_P(\beta), \]
whereas 
\[ v_P(-C\gamma^n-D\delta^n)\geq v_P(C)+nv_P(\gamma). \]
Therefore, for $n>v_P(B/C)/v_P(\gamma/\beta)$, we have that the left-hand side
of (\ref{eq: final, case 1}) is strictly smaller under $v_P$ than the
right-hand side, so (\ref{eq: skolem 4, 1111second}) cannot hold. This bound on $n$
is polynomial in the input size.

Now suppose that there are two roots with strictly smallest valuation with
respect to $v_P$:
\[ 0\leq v_P(\beta)=v_P(\gamma)<v_P(\alpha)=v_P(\delta). \]
Then rewrite (\ref{eq: skolem 4, 1111second}) as
\begin{equation}\label{eq: final, case 2}
B\beta^n\left(\left(-\frac{C}{B}\right)
\left(\frac{\gamma}{\beta}\right)^n-1\right)
=A\alpha^n+D\delta^n.
\end{equation}
Since $\gamma/\beta$ is not a root of unity, the term
$(-C/B)(\gamma/\beta)^n-1$ can be zero for at most one value of $n$. This value
is at most polynomially large in the input size (by Lemma \ref{lem: algebraic
number power problem}).  For all other $n$, we use Theorem \ref{thm:
vdPoorten} to this term. Let $p$ be the unique prime rational integer in the
ideal $P$, and let $d=[\mathbb{K}:\mathbb{Q}]$. Let $H$ be an upper bound for
the heights of $-C/B$ and $\gamma/\beta$. Then by Theorem \ref{thm: vdPoorten}
(van der Poorten), we have
\begin{equation}\label{eq: skolem.vdpapp}
v_P\left(\left(-\frac{C}{B}\right)
\left(\frac{\gamma}{\beta}\right)^n-1\right)
\leq
(48d)^{36}\frac{p^d}{\log p}
(\log H)^2(\log n)^2.
\end{equation}
It is classical that ${\cal N}(P)=p^f$ for some positive integer $f$, so ${\cal
N}(P)\geq p$. Moreover, since $\alpha$ is an algebraic integer, all prime
ideals $P_1,\dots,P_s$ in the factorisation of $[\alpha]$ appear with positive
exponents $k_1,\dots,k_s$:
\[ [\alpha]=P_1^{k_1}\dots P_s^{k_s}. \]
Since ${\cal N}(P_i)\geq2$ for all $P_i$, we have
\[
|{\cal N}_{\mathbb{K}/\mathbb{Q}}(\alpha)|
={\cal N}([\alpha])\geq{\cal N}(P)\geq p.
\]
Therefore, $p$ is at most exponentially large in the input size.  Then we can
write (\ref{eq: skolem.vdpapp}) as
\[
v_P\left(\left(-\frac{C}{B}\right)\left(\frac{\gamma}{\beta}\right)^n-1\right)
\leq E_1(\log n)^2,
\]
where $E_1$ is exponentially large in the input size and independent of $n$.
Now we apply $v_P$ to both sides of equation (\ref{eq: final, case 2}):
\[ v_P(\mathit{LHS})\leq v_P(B)+nv_P(\beta)+E_1(\log n)^2 \]
and
\[ v_P(\mathit{RHS})\geq v_P(A)+nv_P(\alpha). \]
Equation (\ref{eq: skolem 4, 1111second}) cannot hold if
\[ v_P(B)+nv_P(\beta)+E_1(\log n)^2<v_P(A)+nv_P(\alpha), \]
which is implied by
\[ v_P(B/A)+E_1(\log n)^2<n, \]
since $v_P(\alpha)>v_P(\beta)$. Let $E_2=\max\{v_P(B/A),E_1\}$,
then this is implied by
\[ E_2((\log n)^2+1)<n. \]
Since 
\[ (\log n)^2+1<\frac{5\sqrt{n}}{2} \]
for all $n\geq1$, it suffices to have
\[ n>\left(\frac{5}{2}E_2\right)^2. \]
This bound on $n$ is exponential in the size of the input.
\end{proof}